%% file: main.tex
\newif\ifsubmit    
\newif\ifllncs      
\newif\ifexabs      
\newif\ifblind 
\newif\ifacm

\ifllncs
\documentclass[runningheads,a4paper]{llncs}

\else \documentclass[letterpaper,11pt,pdfa]{article}
  \usepackage[in]{fullpage}
\fi
\usepackage{style}

\usepackage{physics}
\renewcommand{\epsilon}{\varepsilon}

\ifsubmit

\else

\fi

\title{A Modular Approach to Succinct Arguments for QMA}
\date{}
\author{James Bartusek\inst{1} \and Jiahui Liu\inst{2} \and Giulio Malavolta\inst{3}}
\ifllncs
\institute{Columbia University, New York, USA \and Fujitsu Research of America, Santa Clara, USA \and Bocconi University, Milan, Italy}
\else
\author{James Bartusek\thanks{Columbia University, New York, USA} \and Jiahui Liu\thanks{Fujitsu Research of America, Santa Clara, USA} \and Giulio Malavolta\thanks{Bocconi University, Milan, Italy}}
\fi

\begin{document}
\maketitle

\begin{abstract}
\input{abstract}
\end{abstract}

\ifllncs\else\newpage\fi
\ifllncs
\else
\tableofcontents
\newpage
\fi

\ifllncs
\input{intro}

\input{succinctifier}
\input{full_protocol}
\input{ack}
\else
\input{intro}
\input{prelim}

\input{remote_state_preparation}

\input{succinctifier}
\input{non_succinct_qma}

\input{full_protocol}
\input{ack}
\fi

\ifllncs\else\newpage\fi

\ifllncs\bibliographystyle{plain}\else\bibliographystyle{alpha}\fi
\bibliography{references, abbrev}

\ifllncs
\else\fi

\end{document}

%% file: abstract.tex
Succinct argument systems are of central importance to modern crytpography, enabling the efficient verification of computational claims. In the classical setting, Kilian (STOC 92) established that any probabilistically checkable proof for NP can be transformed into a succinct argument system for NP using only collision-resistant hash functions. In the quantum setting, recent works have established the feasibility of (classically-verifiable) succinct arguments for QMA, capturing statements that require \emph{quantum} proofs. However, known constructions all rely on the highly structured assumption of learning with errors (LWE), which stands in stark contrast with the unstructured assumptions that suffice for NP.

In this work, we develop a new framework that broadens the cryptographic foundations of succinct arguments for QMA. We assume the existence of (i) an oblivious state preparation (OSP) protocol, which in turn can be constructed from \emph{plain} trapdoor claw-free functions, and (ii) collapsing hash functions, the quantum analogue of collision-resistance. In particular, we obtain the first succinct, classically-verifiable argument system for QMA which does not rely on the hardness of LWE.

Our construction proceeds in two steps. First, we design a \emph{round-efficient} classically-verifiable argument system for QMA based only on the assumption of OSP. Second, we introduce a \emph{generalized communication compression compiler}, which, assuming collapsing hash functions, transforms any $T$-round interactive protocol into one in which the communication size is bounded by $T \cdot \poly(\secp)$ for some fixed $\poly$ independent of the original size of each message. Our compiler extends a quantum rigidity-based communication compression technique of Zhang (QCrypt 25), and may be of independent interest.

%% file: intro.tex
\section{Introduction}

The ability to convince a skeptic of the validity of a mathematical statement in time much less than it takes to \emph{prove} the statement is a widely celebrated result in cryptography \cite{Kilian}, inspiring decades of foundational and applied research on the topic of \emph{succinct argument systems}. Remarkably, Kilian's initial feasibility result utilized only lightweight ``unstructured'' cryptography.  In particular, Kilian showed how to compile any probabilistically checkable proof system (PCP) \cite{10.1145/278298.278306,10.1145/273865.273901} for NP into a succinct argument system for NP using collision-resistant hash functions. 

More recently, the growing interest and investment in quantum information processing has given rise to natural related questions: Can we efficiently verify statements that require \emph{quantum} proofs? And under what cryptographic assumptions? 

Intriguingly, a quantum analogue of Kilian's compiler was worked out by \cite{10.1145/3564246.3585198}, establishing that PCPs for QMA can be compiled into succinct argument systems for QMA, utilizing only unstructured cryptography (in fact, utilizing assumptions that are plausibly even weaker than one-way functions). Unfortunately, it is unknown whether PCPs for QMA exist! Indeed, fundamental properties of quantum information such as unclonability suggest that quantum PCPs may not exist at all (see e.g.\ \cite{aharonov2013quantumpcpconjecture}).

Despite this gap, the community has still managed to demonstrate the feasibility of succinct arguments for QMA \cite{BKL+,metger2024succinct,GKNV25,cryptoeprint:2023/1492} without requiring the existence of quantum PCPs. In fact, these protocols are even \emph{classically-verifiable}, allowing classical clients to verify the correctness of large delegated quantum computations. Unfortunately, the cryptographic structure required to instantiate these protocols is highly specific. As far as we currently understand, each relies specifically on the quantum hardness of learning with errors (LWE):

\begin{itemize}
    \item \cite{BKL+}, \cite{GKNV25}, and \cite{cryptoeprint:2023/1492} build on top of Mahadev's measurement protocol \cite{Mahadev2018}, which requires a trapdoor claw-free function (TCF) with the \emph{adaptive hard-core bit} property. The only known instantiation of this object relies on LWE. 
    \item \cite{metger2024succinct} take a different approach, based on the compiled non-local games paradigm \cite{kalai2023quantum}, but make crucial use of compact quantum homomorphic encryption \cite{MahadevQFHE}, which is also only known assuming the hardness of LWE.\ifllncs\else\footnote{We note that the work of \cite{10.1007/978-3-031-68382-4_8} constructs quantum FHE from post-quantum indistinguishability obfuscation (iO) plus assumptions on cryptographic group actions. However, we don't know to base post-quantum iO on standard cryptographic assumptions, and anyway the leading candidates we do have are lattice-based. }\fi
\end{itemize}

This state of affairs leaves us in an undesirable situation that contrasts sharply with the classical setting: While classical statements can be succintly argued given \emph{any} collision-resistant hash function, we only know succinct arguments for quantumly-provable statements from the conjectured quantum hardness of LWE. 

\subsection{Results}

In this work, we introduce a novel framework for building (classically-verifiable) succinct arguments for QMA, and leverage our approach to diversify the sources of cryptographic hardness under which succinct arguments for QMA exist. In more detail, we utilize the following objects.

\begin{itemize}
    \item \textbf{Oblivious state preparation.} OSP, recently introduced by \cite{bartusek2025power}, is an interactive protocol between a classical client and a quantum server that prepares a BB84 state on the server's side without revealing the choice of basis to the server. It can be constructed from any ``plain'' trapdoor claw-free function (without assuming extra properties such as dual-mode or the adaptive hard-core bit), and is known from LWE \cite{BCMVV}, cryptographic group actions \cite{AMR22}, and the lattice isomorphism problem \cite{cryptoeprint:2025/993}. 
    \item \textbf{Collapsing hash functions.} Collapsing is the post-quantum analogue of collision-resistance, and is a standard ``unstructured'' hardness assumption that holds in the quantum random oracle model \cite{unruh2016computationally}, or from a variety of computational assumptions \cite{zhandry2022new}. Collapsing hash functions are known to imply post-quantum succinct arguments for NP (by adapting the proof of Kilian's protocol to the post-quantum setting) \cite{CMSZ}, though it remains a fascinating open problem whether collapsing hash functions alone (or other sources of unstructured hardness) suffice to build succinct arguments for QMA.
\end{itemize}

Our main result can be summarized as follows.

\begin{theorem}[Informal]\label{thm:informal} There exists a fixed polynomial $\poly$ such that, assuming OSP and collapsing hash functions, there exist a classically-verifiable argument system for QMA with communication size $\poly(\secp)$ and verifier complexity $|x| \cdot \poly(\secp)$, where $|x|$ is the instance size.\ifllncs\else\footnote{We note that some prior work (see e.g.\ definitions in \cite{BKL+,metger2024succinct}) actually obtain a client complexity with \emph{additive} overhead $\tilde{O}(|x|) + \poly(\secp)$ in the instance size, whereas our overhead is multiplicative. However, our protocol still satisfies the fundamental succinctness property that $\poly$ is fixed and independent of the witness size.}\fi
\end{theorem}

We note that, like previous work on classically-verifiable succinct arguments for QMA \cite{BKL+,metger2024succinct,GKNV25,cryptoeprint:2023/1492}, our result is interesting even when applied to BQP languages, establishing that (bounded-error) polynomial-time quantum computations can be classically verified in time much less than the time it takes to execute the quantum computation. 

Our main result can also be viewed in light of recent work \cite{bartusek2025power,BKM+25,ZhaRSP} establishing that either OSP or ``plain'' trapdoor claw-free functions alone are enough to obtain (non-succinct) classical verification of quantum computation. We show that making use of collapsing hash functions along with OSP suffices to realize the additional property of \emph{succinctness}.

Finally, it is straightforward to extend our methods to obtain a succinct protocol for \emph{quantum sampling problems} \cite{EC:CLLW}, underscoring the generality of our approach. In particular, we establish the following.

\begin{theorem}[Informal]
    There exists a fixed polynomial $\poly$ such that, assuming LWE, for any inverse-polynomial $\epsilon$, there exists a classically-verifiable argument for quantum sampling with $\epsilon$-soundness, communication size $\poly(\secp,1/\epsilon)$, and verifier complexity $|Q| \cdot \poly(\secp,1/\epsilon)$, where $|Q|$ is the (potentially succinct) description of the quantum sampler.
\end{theorem}

We note that this statement requires assuming the hardness of LWE but, to the best of our knowledge, succinct arguments for sampling problems were not known from any assumption prior to this work. {We also mention that, as a byproduct of our techniques, we obtain a general compiler that allows us to compress the communication and the verifier-complexity of \emph{any} quantum interactive protocol to some fixed polynomial in the security parameter, regardless of the complexity of the original interactive protocol.}

\section{Techniques}

Our approach is quite modular, and can be broken into two central steps.

\begin{enumerate}
    \item Step I: First, we design a $\poly(\secp)$-round\footnote{Where the polynomial is independent of the QMA witness size.} classically-verifiable protocol for QMA. While there already exist even \emph{constant}-round classically-verifiable protocols in the literature, we obtain for the first time a $\poly(\secp)$-round protocol based only on the assumption of OSP.
    
    \item Step II: Second, we develop a ``generalized communication compression compiler'', and apply it to our $\poly(\secp)$-round protocol in order to obtain a protocol with $\poly(\secp)$ \emph{total bits of communication}. Our compiler, which we believe is of independent interest, is inspired by and extends a recent work of Zhang \cite{cryptoeprint:2023/1492}.
\end{enumerate}
In more detail, we establish the following result. Fix any $\varepsilon = 1/\poly(\secp)$. Then assuming collapsing hash functions, \emph{any} interactive protocol that takes place between a quantum server and a classical client with single-bit output (accept or reject) can be compiled into a protocol with communication that grows only with the original round complexity and a fixed $\poly(\secp, 1/\varepsilon)$. The probability that a malicious server can convince the verifier to accept will only increase by at most $\varepsilon$. 
    
This ability to generically shrink the communication relies crucially on quantum ``rigidity'' techniques, and introduces the need for the client to prepare and send a (bounded polynomial) number of claw states $\ket{x_0} + \ket{x_1}$. Additionally assuming OSP (again), we are able to replace this quantum communication with classical communication, restoring the classical verifiability of our final succinct argument for QMA.

In what follows, we provide more detail about each of these main steps.

\subsection{Round-Efficient Classical Verification for QMA}

\paragraph{Background.} Our first, intermediate, goal en route to succinct arguments for QMA is the following: We aim to construct a classically-verifiable $\poly(\secp)$-round protocol, but where the amount of communication in each round may be non-succinct. That is, we focus first on optimizing for \emph{round} complexity.

By now, there exist a few approaches to classical verification of quantum computation, so we begin by reviewing some of these approaches, and explain why known results fall short of our goals.

\begin{itemize}
    \item The original approach of \cite{Mahadev2018} is centered around the idea of a ``measurement protocol'', and in fact it is known how to leverage this approach to obtain a \emph{constant}-round protocol \cite{ACGH}. However, the underlying measurement protocol requires the use of a TCF with the adaptive hard-core bit property, which we aim to avoid in order to diversify the set of assumptions from which our approach can be instantiated.
    \item More recently, the framework of compiled non-local games \cite{kalai2023quantum} has been successfully utilized to obtain classical argument systems for QMA \cite{natarajan2023bounding}, even from more-general assumptions such as OSP \cite{bartusek2025power} or plain TCFs \cite{BKM+25}. Unfortunately, both of these instantiations suffer from round complexity that depends on the size of the QMA witness. This is due to the fact that they compile a non-local game with very small completeness-soundness gap, and thus inherit a prohibitively large round complexity blow-up from sequential repetition.
\end{itemize}

\paragraph{Our protocol.} We follow the non-local games approach, but begin with a game that already enjoys \emph{constant} completeness-soundness gap, so that sequential repetition will not be too costly.

In more detail, we begin with a simplified version of the two-round two-prover non-local game from \cite{metger2024succinct}, where both questions and answers may grow with the size of the witness.\ifllncs\else\footnote{In particular, we do not require their ``question-reduction'' machinery of sampling questions from $\epsilon$-biased sets and subsampling the Hamiltonian terms.}\fi However, we require that this game satisfies two particular properties.
    \begin{enumerate}[i)]
        \item Alice's operation has low depth. That is, the quantum circuit that operates on the quantum witness, her halves of EPR pairs, and her question, and outputs her answer, has depth that is bounded by some fixed polynomial $\poly(\secp)$ of the security parameter. We observe that this holds for (our simplified version of) \cite{metger2024succinct}'s protocol.
        \item The game is sound against \emph{computationally} non-local strategies \cite{bartusek2025power}. That is, soundness holds against two-party strategies that may share the same Hilbert space but where the left-over state after Alice's operation is computationally independent of her question. We confirm that this holds for \cite{metger2024succinct}'s protocol, as their proof in the ``compiled'' setting only makes use of the computational indistinguishability of Alice's questions.
    \end{enumerate}

We then compile the two-prover game into a single-prover protocol using a ``generalized'' KLVY compiler \cite{kalai2023quantum,bartusek2025power,BKM+25}, which performs Alice's operation under a \emph{blind} classical delegation of quantum computation protocol, followed by Bob's operation in the clear. Blind classical delegation of quantum computation can be constructed from OSP \cite{bartusek2025power}, where the number of rounds grows only with the \emph{depth} of the delegated operation. Thus, our compiled protocol still has a strict $\poly(\secp)$ number of rounds (though the communication in each round may grow with the size of the witness).\ifllncs\else We refer to the reader to \cref{def:compiled-game} for further details.\fi

\subsection{A Generic Approach to Communication Compression}

At the heart of our approach to succinctness is an inherently quantum technique for constructing simulation-based secure computation with succinct communication. The basic idea was pioneered by Zhang \cite{cryptoeprint:2023/1492}, and we extend his work to obtain our general compiler. 

\paragraph{Background.} \cite{cryptoeprint:2023/1492} addresses the following setting. Let $f: \{0,1\}^n \to \{0,1\}^m$ be a function with $m \gg n$, and suppose that a client wants to send $f(x)$ to a server for a random $x \gets \{0,1\}^n$, without revealing the value of $x$. In the fully classical setting, it is provably \emph{impossible} to obtain a succinct (i.e.\ much less than $m$) communication protocol for this functionality with simulation-security \cite{HW}. Interestingly, \cite{cryptoeprint:2023/1492} shows that quantum ``rigidity'' techniques can be applied to obtain a protocol with succinct \emph{quantum} communication.

The basic, simplified, idea is as follows. The client first samples random $x_0,x_1,r_0,r_1 \gets \{0,1\}^n$ and sends the claw state \[\frac{1}{\sqrt{2}}(\ket{0,x_0,r_0} + \ket{1,x_1,r_1})\] to the server. An honest server will first compute the function $f$ on the input register in superposition, yielding \[\frac{1}{\sqrt{2}}(\ket{0,x_0,r_0,f(x_0)} + \ket{1,x_1,r_1,f(x_1)}),\] and then ``delete'' the input by measuring the ``$x$'' register in the Hadamard basis to obtain a string $c \in \{0,1\}^n$. Then, the client selects one of three options.

\begin{itemize}
    \item \textbf{Standard basis test:} The server measures the remaining registers in the standard basis to obtain the (large) string $(b,r_b,f(x_b))$. It then engages in a succinct ``commit-reveal-prove'' protocol (based on a collapsing hash function) with the client to argue that this information was honestly generated. 
    \item \textbf{Hadamard basis test:} The server measures the remaining registers in the Hadamard basis to obtain the (large) string $d$. It then engages in a similar succinct ``commit-reveal-prove'' with the client to prove honest behavior.
    \item \textbf{Output generation:} The server measures the first two remaining registers in the standard basis to obtain the (short) string $(b,r_b)$, which it sends back to the client. The server's state collapses to $f(x_b)$, while the client knows $x_b$.
\end{itemize}

The standard and Hadamard basis tests allow \cite{cryptoeprint:2023/1492} to make a rigidity argument on any (successful) server's system, enabling a quantum simulation strategy that has no classical counterpart. Using garbled circuits, \cite{cryptoeprint:2023/1492} then shows how to extend the above ``random-input'' functionality to a ``chosen-input'' functionality, where the client is able to choose any value of $x$, and delegate $f(x)$ to the server.

\paragraph{Committed secure function evaluation.} In this work, we adapt the basic idea outlined above to the following setting. Consider any protocol that takes place between a (potentially quantum) server and a classical client, where the goal of the server is to try to convince the client to accept the interaction. For example, any classical interactive argument for QMA fits this description. We aim to replace each round of communication in such a protocol with a fixed $\poly(\secp)$-communication sub-protocol, regardless of how many bits were communicated in the round. 

Prior work \cite{BKL+} has shown how to generically compress the server's answers using succinct commitments to these messages followed by (state-preserving) arguments of knowledge \cite{LMS}. In order to also compress the client's messages, we propose the following high-level strategy.

\begin{itemize}
    \item The client samples random coins $s \gets \{0,1\}^\secp$ that it will use throughout the protocol. Note that this randomness can be expanded using a PRG, and thus the remaining client's strategy can be taken to be completely deterministic. In particular, the client's \emph{next message} function in each round $i$ can be expressed as $f_i(s,m_1,\dots,m_i)$, where $m_1,\dots,m_i$ are the prover's previous messages.
    \item Each of the prover messages are replaced with \emph{succinct commitments} $\Com(m)$.
    \item Each of the verifier messages $f_i(s,m_1,\dots,m_{i-1})$ are computed via a succinct ``committed secure function evaluation'' protocol with public inputs $\Com(m_1),\dots,\Com(m_{i-1})$, prover inputs $m_1,\dots,m_i$ and (private) verifier input $s$. This protocol is an extension of \cite{cryptoeprint:2023/1492}'s secure function evaluation protocol described above, where we additionally allow for \emph{large} server inputs that have previously been \emph{succinctly committed}. 
\end{itemize}

As one of our main technical contributions, we develop the required succinct committed secure function evaluation protocol, by extending the protocol of \cite{cryptoeprint:2023/1492} to handle (committed) prover inputs, and provide a new self-contained proof of security. We refer to \cref{sec:compression_compiler} for further details.

\paragraph{Restoring classical communication.} Finally, we replace the quantum claw-state communication with classical communication, again relying on the assumption of OSP. We do this via two steps adapted from the prior work of \cite{ZhaRSP}\ifllncs\else, and worked out in \cref{sec:rspv}\fi. 

\begin{enumerate}
    \item First, \cite{ZhaRSP} argues that (plain) trapdoor claw-free functions imply a notion of \emph{verifiable} state preparation (VSP) for BB84 states, satisfying a simulation-based notion of security. This strategy can be adapting to showing the OSP implies VSP for BB84 states.
    \item Second, \cite{ZhaRSP} shows that VSP for BB84 states implies VSP for claw states, which is what we require to instantiate the above with only classical communication.
\end{enumerate}


%% file: succinctifier.tex
\section{Generalized Communication Compression Compiler}
\label{sec:compression_compiler}

In the following we present our compiler to generically reduce the communication complexity of interactive protocols between QPT provers and PPT verifiers. For an interactive protocol, we denote its communication complexity as the number of bits exchanged during the interaction. In this section, we will prove the following main theorem.

\begin{theorem}[Communication Compression]\label{thm:compcom}
    Let $z \gets \langle P(1^\secp, \chi),V(1^\secp ,\chi)\rangle$ be an interactive protocol between a QPT prover $P$ and PPT verifier $V$ with public input $\chi \in X$ (which, in particular, completely determines the set of honest computations to be performed throughout the protocol), verifier output $z \in Z \cup \{\bot\}$, where output $\bot$ indicates an abort, and $r= r(\lambda)$ rounds. Assuming the existence of succinct collision-resistant commitments\ifllncs\else~(\cref{sec:hash-functions}) \fi and state-preserving succinct AoKs\ifllncs\else~(\cref{def:AoK})\fi, there exists an interactive protocol $z \gets \langle \Tilde{P}(1^\secp, \chi, \varepsilon),\Tilde{V}(1^\secp ,\chi, \varepsilon)\rangle$ where $z \in Z \cup \{\bot\}$ such that the following holds.
    \begin{itemize}
        \item \textbf{Completeness.} For any $\chi \in X$ and any $\varepsilon = 1/\poly(\secp)$, \[\mathsf{TV}\left(\left\{z \gets \langle P(1^\secp,\chi),V(1^\secp,\chi)\rangle\right\},\left\{z \gets \langle \tilde P(1^\secp,\chi, \varepsilon),\tilde V(1^\secp,\chi ,\varepsilon)\rangle\right\}\right) = \negl(\secp).\]
        \item \textbf{Soundness.}  For any $\chi \in X$, $\varepsilon = 1/\poly(\secp)$, and QPT adversary $\{\tilde\Adv_\secp\}_{\secp \in \bbN}$, there exists a QPT adversary $\{\Adv_\secp\}_{\secp \in \bbN}$ such that for any QPT distinguisher $\{\D_\secp\}_{\secp \in \bbN}$,
        \[\abs{\Pr\left[1 \gets \D_\secp \left(\langle \cA_\secp,V(1^\secp,\chi)\rangle\right)\right]  - \Pr\left[1 \gets \D_\secp\left(\langle \tilde\Adv_\secp,\tilde V(1^\secp,\chi, \varepsilon)\rangle\right)\right]}\leq O(\varepsilon).\]
        \item \textbf{Efficiency.} There exists a fixed polynomial $p(\secp) = \poly(\secp)$ such that the communication complexity of the protocol is bounded by 
        $r \cdot p(\secp, 1/\varepsilon)$ and the runtime of the verifier is bounded by $r \cdot |\chi| \cdot p(\secp, 1/\varepsilon)$. 
    \end{itemize}

    
    Moreover, additionally assuming OSP\ifllncs\else~(\cref{def:OSP})\fi, there exists a compressed protocol where the verifier $\tilde{V}$ is completely classical.  
\end{theorem}
We note that the notion of soundness as defined above immediately implies that any bound that one can establish for the non-compiled protocol immediately carries over to the compiled protocol, up to arbitrarily small inverse-polynomial factors.

\subsection{Definition: Committed Secure Function Sampling}

A Secure Function Sampling (SFS) protocol \cite{cryptoeprint:2023/1492} allows a prover and a verifier to interactively sample a random input/output pair for a given function. Here we define an extension of this primitive, where the function takes an additional input that corresponds to the pre-image of a commitment. We assume that the commitment is \emph{deterministic} and satisfies the notion of collision-resistance\ifllncs\else~(see \cref{sec:hash-functions})\fi.

\begin{definition}[Committed SFS]\label{def:ComSFS}
Committed Secure Function Sampling (SFS) for a function $f:\{0,1\}^{n} \times\{0,1\}^{\ell} \to\{0,1\}^m$ and a commitment $\mathsf{Com}$ is a protocol that takes place between a PPT verifier $V$ and a QPT prover $P$:
 \[(z,x) \gets \mathsf{cSFS}\langle P(1^\secp,\epsilon, y),V(1^\secp,\epsilon, c)\rangle,\] where $\epsilon$ is an error parameter, $z$ is the prover's output, and $x\in \{0,1\}^n \cup \{\bot\}$ is the verifier's output. It should satisfy the following properties.
\begin{itemize}
\item \textbf{Correctness.} For any $\epsilon = 1/\poly(\secp)$ and all pairs $(y, c = \mathsf{Com}(y))$, \[\TD\left(\mathsf{cSFS}\langle P(1^\secp,\epsilon, y),V(1^\secp,\epsilon,c)\rangle,\frac{1}{2^n}\sum_{x \in \{0,1\}^n}\ketbra{x} \otimes \ketbra{f(x,y)}\right) = \negl(\secp).\]
\item \textbf{Security.} For any $\epsilon = 1/\poly(\secp)$, any QPT environment $\cE = \{\cE_\secp\}_{\secp \in \bbN}$, and any QPT adversarial prover $\cA = \{\cA_\secp\}_{\secp \in \bbN}$, there exists a QPT simulator $\Sim = \{\Sim_\secp\}_{\secp \in \bbN}$ such that for any QPT distinguisher $\cD = \{\cD_\secp\}_{\secp \in \bbN}$, the following holds. Define the following games.\\

        \noindent\underline{$\ExpREAL[\cE,\cA,\epsilon]$}
        \begin{itemize}
            \item Run $y,\rho_{A,E} \gets \cE(\Com)$.
            \item Run $(A,x) \gets \mathsf{cSFS}\langle\cA(A), V(1^\secp,\epsilon, \mathsf{Com}(y))\rangle$.
            \item Output $(x,A,E)$.
        \end{itemize}

        \noindent\underline{$\ExpIDEAL[\cE,\Sim,\epsilon]$}
        \begin{itemize}
            \item Run $y,\rho_{A,E} \gets \cE(\Com)$.
            \item Run $(b,A) \gets \Sim(1^\secp,\epsilon,A)$.
            \item If $b = 1$, sample $x \gets X$ and set $z =f (x,y)$, otherwise if $b = 0$, set $x = z = \bot$.
            \item Run $A \gets \Sim(1^\secp,\epsilon,A,z)$.
            \item Output $(x,A,E)$.
        \end{itemize}
Then,
\[\bigg| \Pr[\cD(\ExpREAL[\cA,\epsilon,\rho, y]) = 1] - \Pr[\cD(\ExpIDEAL[\Sim,\epsilon,\rho,y ]) = 1]\bigg| \leq \epsilon.\]
\item \textbf{Efficiency.} There exists a fixed polynomial $p(\secp) = \poly(\secp)$ such that the communication complexity of the protocol is bounded by $p(\secp,1/\epsilon,n)$ and the runtime of the verifier is bounded by $|f| \cdot p(\secp, 1/\varepsilon,n)$, where $|f|$ is the \emph{description size} of $f$ (as opposed to the time it takes to execute). 
\end{itemize}
\end{definition}
\paragraph{From Random-Input to Chosen-Input.} A variant of the above primitive allows the verifier to choose an arbitrary input $x$, while attaining the same security guarantees. To differentiate between these two primitives we refer to them as \emph{random-input} committed SFS and \emph{chosen-input} committed SFS. It is shown in \cite{cryptoeprint:2023/1492} that random-input SFS and chosen-input SFS are equivalent assuming the existence of one-way functions. The same transformation works in the committed settings, which we sketch here. 

Let $g(\cdot, y)$ be the function that takes as input some random coins $r \in\{0,1\}^\lambda$ and returns a garbling \cite{yao1986generate} of the circuit computing $f_y(\cdot ) = f(\cdot , y)$. Running a random-input committed SFS results in the verifier receiving a random $r$ and the prover receiving the garbled circuit $g(r, y) = \Tilde{f}_y$. The verifier can then expand $r$ into pairs of labels and, for each pair, send to the prover the $x_i$-th label. Security follows by first invoking \cref{def:ComSFS} to program the output of the prover and then substituting it with a simulated garbled circuit, which is indistinguishable by the security of garbled circuits. Since this is a standard transformation and it is already proven (up to trivial syntactical modifications) in \cite{cryptoeprint:2023/1492}, we omit a formal proof and we directly state the following lemma.
\begin{lemma}[Random-input implies chosen-input committed SFS]\label{lmm:random-chosen}
    Assuming one-way functions, the existence of random-input committed SFS implies the existence of chosen-input committed SFS.
\end{lemma}


\subsection{Random-Input Committed SFS}\label{sec:SFSprotocol}

We present a self-contained description of our random-input committed SFS protocol, which is itself a variant of the protocol proposed in \cite{cryptoeprint:2023/1492}. We assume the existence of the following ingredients:
\begin{itemize}
    \item A state-preserving succinct AoK for NP\ifllncs\else~(\cref{sec:AoK})\fi.
    \item A succinct commitment $\mathsf{Com}$ satisfying collision resistance\ifllncs\else~(\cref{sec:hash-functions})\fi, whose output size is a fixed $\poly(\secp)$.
\end{itemize}
We describe the protocol in the settings where the verifier is entirely classical, except that in the first round it sends a \emph{claw state} to the prover. A claw state is a state of the form
\[
\frac{\ket{0,z_0} + \ket{1,z_1}}{\sqrt{2}}
\]
for randomly sampled $z_0, z_1 \gets \{0,1\}^{\poly(\secp)}$. \ifllncs\else Looking ahead, in \cref{sec:rspv} we shall see how to make the verifier entirely classical, but to keep things simple, here we present the protocol with explicit claw states sent by the verifier.\fi It is a straightforward consequence of the \emph{distinguishing implies swapping} principle \cite{AAS} (see also \cite{HMT23,MW24}) that the following states are negligibly close in trace distance
\begin{equation}\label{eq:distswap}
 \frac{\ket{0,z_0}+\ket{1,z_1}}{\sqrt{2}}\frac{\bra{0,z_0}+\bra{1,z_1}}{\sqrt{2}} \approx_{\negl(\secp)} \frac{\ketbra{0,z_0} + \ketbra{1,z_1}}{2}
\end{equation}
that is, a claw state is close to the classical mixture over each claw. Naturally, \cref{eq:distswap} implies that it is hard for any (possibly unbounded) algorithm to return both $(z_0, z_1)$ given a claw state, with non-negligible probability.

The protocol is parametrized by a function $f:\{0,1\}^n\times\{0,1\}^\ell\to \{0,1\}^m$ and a commitment scheme $\mathsf{Com}$. In addition to the public description of $f$ and a parameter $\varepsilon$, the prover is given a string $y \in\{0,1\}^\ell$ and the verifier is given $\mathsf{Com}(y) = c$.

We describe the following protocol between the prover and the verifier.

\begin{center}
\begin{longfbox}[breakable=true, padding=1em, padding-right=1.8em, padding-top=1.2em, margin-top=1em, margin-bottom=1em]
%
\noindent \underline{Random-input Committed SFS}

\textbf{Parameters:}\text{ $p$ (the probability of a test round), to be chosen later}\\

\textbf{Input:}\text{ $f:\{0,1\}^n\times\{0,1\}^\ell\to \{0,1\}^m$, $y$}
\begin{enumerate}
    \item The verifier sends to the prover a random claw state
    \[
    \frac{1}{\sqrt{2}}\left(\ket{0,z_0} + \ket{1,z_1}\right) =
        \frac{1}{\sqrt{2}}\left(\ket{0,x_0, r_0, s_0} + \ket{1,x_1, r_1, s_1}\right)
    \]
    where $z_b = (x_b, r_b, s_b) \in \{0,1\}^n \times\{0,1\}^\lambda \times\{0,1\}^\lambda$.
    \item The prover coherently computes
    \[
         \frac{1}{\sqrt{2}}\left(\ket{0,x_0, r_0, s_0, f(x_0, y)} + \ket{1,x_1, r_1, s_1, f(x_1, y)}\right)
    \]
    then it measures the second and third registers in the Hadamard basis to obtain $c \in\{0,1\}^n$ and $d \in \{0,1\}^\lambda$. The residual post-measurement state is 
    \[
    \frac{1}{\sqrt{2}}\left((-1)^{\theta_0}\ket{0, s_0, f(x_0, y)} + (-1)^{\theta_1}\ket{1, s_1, f(x_1, y)}\right)
    \]
    where
    \begin{align*}        
    \theta_0 = \langle x_0, c \rangle \oplus \langle r_0, d \rangle  \quad\text{and}\quad
    \theta_1 = \langle x_1, c \rangle \oplus \langle r_1, d \rangle.
    \end{align*}
    The prover sends $c,d$ back to the verifer, who checks that $d \neq 0^\lambda$ and recomputes the angles $\theta_0, \theta_1$.
    \item The verifier sample a biased coin and with probability $p$ enters a test round and afterwards starts over from the first step. Otherwise, with probability $1-p$, it enters a measurement round.
    \item (Test Round) The verifier selects between the following two tests with probability $1/2$.
    \begin{itemize}
    \item (Computational Test) The prover measures its state in the computational basis to obtain an outcome $(u,v,w)\in \{0,1\} \times \{0,1\}^\lambda \times \{0,1\}^m$. The prover sends to the verifier $\mathsf{Com}(u,v,w) = \Tilde{c}$ and engages in an AoK with the verifier for knowledge of the committed string. Afterwards, the verifier reveals $(s_0, s_1, x_0, x_1)$ and the prover engages in an AoK proving that:
    \begin{itemize}
        \item The string $(u,v,w)$ is the pre-image of the commitment $\tilde{c}$, i.e., $\mathsf{Com}(u,v,w) = \Tilde{c}$.
        \item The string $y$ is the pre-image of $c$, i.e., $c = \mathsf{Com}(y)$.
        \item It holds that $(u,v,w) = (0, s_0, f(x_0, y))$ or \ifllncs
        \\ \else\fi $(u,v,w) = (1, s_1, f(x_1, y))$.
    \end{itemize}
    
    \item (Hadamard Test) The prover measures its state in the Hadamard basis to obtain an outcome $(u,v,w)\in \{0,1\} \times \{0,1\}^\lambda \times \{0,1\}^m$. The prover sends to the verifier $\mathsf{Com}(u,v,w) = \Tilde{c}$ and engages in an AoK with the verifier for knowledge of the committed string. Afterwards, the verifier reveals $(\theta_0, \theta_1, s_0, s_1, x_0, x_1)$ and the prover engages in an AoK proving that:
    \begin{itemize}
        \item The string $(u,v,w)$ is the pre-image of the commitment $\tilde{c}$, i.e., $\mathsf{Com}(u,v,w) = \Tilde{c}$.
        \item The string $y$ is the pre-image of $c$, i.e.,  $c = \mathsf{Com}(y)$.
        \item It holds that $\theta_0 \oplus \theta_1 = u \oplus \langle v, s_0 \oplus s_1 \rangle \oplus \langle w, f(x_0, y) \oplus f(x_1, y) \rangle$.
    \end{itemize}
    \end{itemize}
    \item (Measurement Round) The prover measures its state in the computational basis to obtain an outcome $(u,v,w)\in \{0,1\} \times \{0,1\}^\lambda \times \{0,1\}^m$. The prover outputs $w$ and sends $(u,v)$ to the verifier. If $(u,v) = (0,s_0)$ the verifier returns $x_0$ and if $(u,v) = (1,s_1)$ the verifier returns $x_1$. Otherwise the verifier aborts.
\end{enumerate}
\end{longfbox}
\end{center}

\subsection{Analysis}

The correctness of the protocol is straightforward, so in this section we analyze its security. We say that a protocol is $\beta$-secure if the distinguishing advantage in the security game is bounded by $\beta$. Then it suffices to prove the following lemma.

\begin{lemma}[$\beta$-Security]\label{lmm:beta-security}
    Assume the existence of succinct collision-resistant commitments and state-preserving succinct AoKs.  
    Fix any $\delta = 1/\poly(\secp)$ and consider a QPT prover $\mathcal{A}$ that passes all tests with probability $1 - O(\delta)$. Then the protocol in \cref{sec:SFSprotocol} is $\beta$-secure against $\mathcal{A}$, with $\beta = O(\sqrt{\delta})$.
\end{lemma}
Following arguments made in \cite[Section 6.1]{cryptoeprint:2023/1492},  \cref{lmm:beta-security} suffices to establish the security of the protocol, since it implies that for any $\varepsilon = 1/\poly(\secp)$, we can pick a sufficiently large $p$ such that the protocol is $\varepsilon$-secure. Thus, we concentrate on proving \cref{lmm:beta-security}. Our proof largely follows the strategy in \cite{cryptoeprint:2023/1492}, with some modifications to handle the committed inputs.

Before presenting the formal proof, let us establish some notation. To characterize the joint state between the prover and the verifier, we purify the random choices of the verifier, so prior to the beginning of the execution the joint state between the prover and the verifier is:
\begin{align*}
\ket{\vartheta}_{RVSA}&=\left(\frac{1}{\sqrt{2} \cdot 2^{n + 2\lambda}}\sum_{z_0,z_1} \ket{z_0,z_1}_R \ket{z_0,z_1}_V \left(\ket{0,z_0}_S+\ket{1,z_1}_S\right)\right)\otimes \ket{\alpha}_A\\
&= \frac{1}{\sqrt{2}\cdot 2^{n + 2\lambda}}\left(\sum_{z_0,z_1} \ket{z_0,z_1}_R \ket{z_0,z_1}_V \ket{0,z_0}_S\ket{\alpha}_A+\sum_{z_0,z_1} \ket{z_0,z_1}_R \ket{z_0,z_1}_V \ket{1,z_1}_S\ket{\alpha}_A\right)
\end{align*}
where $R$ is the purification register, $V$ is the internal register of the verifier, $S$ is the state that the verifier sends to the prover, and $\ket{\alpha}$ is some non-uniformly generated auxiliary state that the prover keeps in its internal register $A$. Recall that $z_b = (x_b, r_b,s_b) \in \{0,1\}^{n + 2\lambda}$, then the state of the protocol is given by 
\[
\Tr_R\ket{\vartheta}_{RVSA}
\]
but we can delay the tracing out until the end of the protocol and, by the principle of delayed measurement, this does not change the distribution of the state. Without loss of generality we can think of the first-round attacker as applying a unitary followed by a measurement of two designated sub-registers of $A$ in the computational basis. By linearity, this results into a state of the form
\begin{equation}\label{eq:RVSAprop}
\ket{\varphi}_{RVSA}\propto \sum_{b,c,d,z_0,z_1} \ket{z_0,z_1, c,d,\theta_0, \theta_1}_R \ket{x_0,x_1, s_0,s_1,\theta_0, \theta_1}_V\ket{\varphi_{c,d,z_b, b}}_{SA}
\end{equation}
omitting normalization factors, where $\theta_b=\theta_b(x_b,r_b,c,d)$ and the state $\ket{\varphi_{c,d,z_b, b}}$ is possibly sub-normalized. Note that we have dropped $(r_0, r_1, c, d)$ from the register $V$ (but not from $R$) since the verifier no longer needs that information after computing $(\theta_0, \theta_1)$ and so we can assume without loss of generality that the verifier traces it out.

Next, we proceed with some simple claims about the tests for the state as defined above. For any given $y \in\{0,1\}^\ell$ fixed in the experiment, the following holds for the computational-basis test.

\begin{claim}[Computational-basis Extractor]\label{claim:comp_extractor}
    Let $\ket{\varphi}_{RVSA}$ be the purified joint state of the protocol after the first round for a QPT adversary that passes the computational basis test with probability at least $1- O(\delta)$. There exists a QPT isometry $\mathsf{E}^{(\text{comp})}$ acting on $SA$ plus an additional register $E$ such that
    \[
    \Tr\left( \Pi^{(\text{eq})}_{VE} (\mathsf{Id}\otimes \mathsf{E}^{(\text{comp})}) \ketbra{\varphi}{\varphi}_{RVSA}(\mathsf{Id}\otimes \mathsf{E}^{(\text{comp})})^\dagger \right) \geq 1-O(\delta)
    \]
    where
    \[
    \Pi^{(\text{eq})}_{VE} := \sum_{b,s_0,s_1, x_0, x_1} \ketbra{s_0,s_1,x_0,x_1}{s_0,s_1,x_0,x_1}_V \otimes \ketbra{b,s_b, f(x_b,y)}{b,s_b, f(x_b,y)}_E.
    \]
\end{claim}
\begin{proof}[Proof of \cref{claim:comp_extractor}]
    By the extractability \ifllncs\else (\cref{def:AoK}) \fi of the \emph{first} AoK, there exists an extractor $\mathsf{E}^{(\text{ext})}$ that returns the pre-image of the commitment $\tilde{c}$, and we define $\mathsf{E}^{(\text{comp})}$ to be the purified adversarial unitary for the computational basis test, followed by the purification of such an extractor. It is clear that, running the extractor with parameter $O(\delta)$ and measuring the register $E$ in the computational basis returns a valid pre-image of the commitment $\Tilde{c}$ with probability at least $1 - O(\delta)$. 
      
    All is left to be shown is that such pre-image is indeed of the desired form $(b, s_b, f(x_b,y))$, for some $b\in\{0,1\}$. We claim that this holds with probability at least $1-O(\delta)$. Assume towards contradiction that this is not the case, then consider the following thought experiment, starting from the state $\ket{\varphi}_{RVSA}$:
    \begin{itemize}
        \item Run the first step of the computational basis test, where the adversary returns a commitment.
        \item Run the the extractor for the \emph{first} AoK, which returns an accepting transcript and a valid witness with probability at least $1 - O(\delta)$.
        \item Continue with the execution of the computational basis test, then run the extractor for the \emph{second} AoK. Measure the output of the extractor to obtain a string $(u',v',w')$.
    \end{itemize}
%
    Since passing the first AoK is a necessary condition for passing the computational basis test, conditioning on this event can only increase the probability 
    that the attacker passes the overall computational basis test, which includes a successful run of the second AoK. Furthermore, by the state-preserving extractability, the state produced after running the first extractor is $O(\delta)$-indistinguishable, for any QPT distinguisher, from a real run of the protocol.

    Thus, conditioning on the first extraction succeeding, the probability that the adversary passes the second AoK is at least $1-O({\delta})$. By the extractability \ifllncs\else (\cref{def:AoK}) \fi of the \emph{second} AoK the second extraction succeeds with probability also $1 - O({\delta})$. However, in this event it must be the case that
    \[
    (u,v,w) \neq (u',v',w') \in \{(0,s_0,f(x_0,y)),(1,s_1,f(x_1,y)) \}
    \]
    contradicting the collision resistance \ifllncs\else (\cref{def:crh}) \fi of the commitment scheme. This concludes our proof.
\end{proof}
A similar statement holds for the Hadamard-basis test.
\begin{claim}[Hadamard-basis Extractor]\label{claim:had_extractor}
    Let $\ket{\varphi}_{RVSA}$ be the purified joint state of the protocol after the first round for a QPT adversary that passes the Hadamard basis test with probability at least $1- O(\delta)$. There exists a QPT isometry $\mathsf{E}^{(\text{had})}$ acting on $SA$ plus an additional register $E$ such that
    \[
    \Tr\left( \Pi^{(\text{sub})}_{VE} (\mathsf{Id}\otimes \mathsf{E}^{(\text{had})}) \ketbra{\varphi}{\varphi}_{RVSA}(\mathsf{Id}\otimes \mathsf{E}^{(\text{had})})^\dagger \right) \geq 1-O(\delta)
    \]
    where
    \begin{align*}
    \Pi^{(\text{sub})}_{VE} := &\sum_{s_0,s_1,x_0,x_1, \theta_0, \theta_1} \ketbra{s_0,s_1,x_0,x_1, \theta_0, \theta_1}{s_0,s_1,x_0,x_1, \theta_0, \theta_1}_V \\ &\otimes 
    \sum_{(u,v,w)\in S_{s_0,s_1,x_0,x_1,\theta_0,\theta_1}}
    \ketbra{u,v,w}{u,v,w}_E
    \end{align*}
    and the subspace $S_{s_0,s_1,x_0,x_1,\theta_0,\theta_1}$ consists of all $(u,v,w)$ satisfying \[\theta_0 \oplus \theta_1 = u \oplus \langle v, s_0 \oplus s_1 \rangle \oplus \langle w, f(x_0, y) \oplus f(x_1, y) \rangle.\]
\end{claim}
\begin{proof}[Proof of \cref{claim:had_extractor}]
    The proof follows along the same lines as the proof of \cref{claim:comp_extractor} with the only difference being that the equation relation checked by the test is the membership in the subspace $S_{s_0,s_1,x_0,x_1,\theta_0,\theta_1}$. 
\end{proof}
We are now ready to prove \cref{lmm:beta-security}.

\begin{proof}[Proof of \cref{lmm:beta-security}]
Let $\ket{\varphi}$ be the state after the first round of interaction, as defined in \cref{eq:RVSAprop}. We claim that
\begin{align}\label{eq:uhlmann}
\Tr_R &
\sum_{b,c,d,z_0,z_1} \ket{z_0,z_1, c,d,\theta_0, \theta_1}_R \ket{x_0,x_1, s_0,s_1,\theta_0, \theta_1}_V\ket{\varphi_{c,d,z_b, b}}_{SA}
\nonumber \\
\approx_{O(\sqrt{\delta})} 
\Tr_R& 
\sum_b \sum_{x_{1-b}, s_{1-b}, \theta_{1-b}}\ket{{x_{1-b}, s_{1-b}, \theta_{1-b}}}_R \ket{x_{1-b}, s_{1-b}, \theta_{1-b}}_V  \\ &\otimes 
\sum_{z_b, c, d,\theta_b} \ket{\theta_b, c, d,  \xi_{\theta_b, c, d, z_b}}_R \ket{x_b, s_b, \theta_b}_V \ket{\varphi_{c,d,z_b, b}}_{SA} \nonumber
\end{align}
 for some state $\ket{\xi_{\theta_b, c, d, z_b}}$. We defer the proof of \cref{eq:uhlmann} to a later point in this proof. Since tracing out the purification register does not change the state of the protocol, we can equivalently analyze the state on the RHS, which, up to a re-arrangement of the registers, corresponds to  
\begin{align} \label{eq:real_state}
&\sum_b \sum_{x_{1-b}, s_{1-b}, \theta_{1-b}}\ket{{x_{1-b}, s_{1-b}, \theta_{1-b}}}_R \ket{x_{1-b}, s_{1-b}, \theta_{1-b}}_V 
\nonumber\\&\otimes 
\sum_{\theta_b} \ket{\theta_b}_R \ket{\theta_b}_V
\sum_{z_b, c, d} \ket{c, d,  \xi_{\theta_b, c, d, z_b}}_R \ket{x_b, s_b}_V \ket{\varphi_{c,d,z_b, b}}_{SA}.
\end{align}
 Note in particular that, for each $b\in\{0,1\}$, we can identify a partition of the registers $R$ and $V$ such that the corresponding states are in tensor product.

By \cref{claim:comp_extractor} we know that there exists an isometry $\mathsf{E}^{(\text{comp})}$ acting on the prover registers $SA$ plus an auxiliary register $E$, which returns a valid claw with probability $1 - O({\delta})$. 
Let $\Pi^{(\text{eq})} = \Pi^{(\text{eq},0)} + \Pi^{(\text{eq},1)}$ where
\[
\Pi^{(\text{eq},b)}_{VE} = \sum_{s,x}\ketbra{s,x}{s,x}_V \otimes\ketbra{b,s,f(x,y)}{b,s,f(x,y)}_E
\]
be the projection onto the accepting subspace of the extractor. Then we claim that for $b\in\{0,1\}$ we have
\begin{align}\label{eq:comp_basis_ind}
    &\sum_{\theta_b} \ket{\theta_b}_R \ket{\theta_b}_V
\sum_{z_b, c, d} \ket{c, d,  \xi_{\theta_b, c, d, z_b}}_R \ket{x_b, s_b}_V \mathsf{E}^{(\text{comp})}\ket{\varphi_{c,d,z_b, b}}_{SA}\nonumber \\
\approx_{O({\sqrt{\delta}})}& 
\sum_{\theta_b} \ket{\theta_b}_R \ket{\theta_b}_V
\sum_{z_b, c, d} \ket{c, d,  \xi_{\theta_b, c, d, z_b}}_R \Pi^{(\text{eq},b)}_{VE} \ket{x_b, s_b}_V \mathsf{E}^{(\text{comp})}\ket{\varphi_{c,d,z_b, b}}_{SA}.
\end{align}
This follows by Gentle Measurement \ifllncs\else(\cref{lmm:gentle}) \fi plus a reduction to the hardness of finding a claw: Using the adversarial map of the first round, plus the isometry of the extractor, the reduction can efficiently prepare the state \[\sum_{z_b, c, d} \ket{x_b,s_b}\mathsf{E}^{(\text{comp})}\ket{\varphi_{c,d,z_b, b}},\] without any knowledge of $z_{1-b}$. Then, if the projection $\Pi^{(\text{eq},1-b)}$ succeeds, the reduction recovers $(1-b,s_{1-b}, f(x_{1-b},y))$, contradicting the security of the claw state, see \cref{eq:distswap}.

To summarize, we have characterized the state of the adversary right after the first round of interaction (i.e., right after sending $(c,d)$ to the verifier) as
\begin{align}\label{eq:state}
&\sum_b \left( \sum_{x_{1-b}, s_{1-b}, \theta_{1-b}}\ket{{x_{1-b}, s_{1-b}, \theta_{1-b}}}_R \ket{x_{1-b}, s_{1-b}, \theta_{1-b}}_V \right)\nonumber
\\&\otimes 
\left(
\sum_{\theta_b} \ket{\theta_b}_R \ket{\theta_b}_V
\sum_{z_b, c, d} \ket{c, d,  \xi_{\theta_b, c, d, z_b}}_R \Pi^{(\text{eq},b)}\ket{x_b, s_b}_V \mathsf{E}^{(\text{comp})}\ket{\varphi_{c,d,z_b, b}}_{SA}
\right)\nonumber\\   
=&\sum_b  \sum_{\theta_{1-b}}\ket{\theta_{1-b}}_R\ket{\theta_{1-b}}_V \otimes  \sum_{\theta_b} \ket{\theta_b}_R \ket{\theta_b}_V \otimes  \nonumber
\\&
\underbrace{\sum_{x_{1-b}, s_{1-b}}\ket{{x_{1-b}, s_{1-b}}}_R \ket{x_{1-b}, s_{1-b}}_V \otimes \sum_{z_b, c, d} \ket{c, d,  \xi_{\theta_b, c, d, z_b}}_R \Pi^{(\text{eq},b)}\ket{x_b, s_b}_V \mathsf{E}^{(\text{comp})}\ket{\varphi_{c,d,z_b, b}}_{SA}}_{=: \ket{\varphi_{b, \theta_b}}}\nonumber\\
=&\sum_b  \sum_{\theta_{1-b}}\ket{\theta_{1-b}}_R\ket{\theta_{1-b}}_V \otimes  \sum_{\theta_b} \ket{\theta_b}_R \ket{\theta_b}_V \otimes \ket{\varphi_{b, \theta_b}}_{RVSAE}
\end{align}
where the above equalities are just re-arrangements of registers. Note that the state $\ket{\varphi_{b, \theta_b}}$ is independent of $\theta_{1-b}$. We claim that
\begin{equation}\label{eq:phi_equality}
    \ket{\varphi_{0, 0}} \approx_{O(\sqrt{\delta})}  - \ket{\varphi_{0, 1}}
    \quad\text{and}\quad
    \ket{\varphi_{1, 0}} \approx_{O(\sqrt{\delta})}   -\ket{\varphi_{1, 1}}
\end{equation}
and we defer the proof of \cref{eq:phi_equality} to a later point in this proof.

We henceforth assume that the state of the prover is as in \cref{eq:state} at the beginning of the second round of interaction. This is without loss of generality: The adversary can apply $(\mathsf{E}^{(\text{comp})})^\dagger$ (which is efficiently implementable) to recover a state which, by \cref{eq:comp_basis_ind}, is within $O({\sqrt{\delta}})$ trace distance from the one in \cref{eq:real_state}.

By \cref{claim:had_extractor} we know that there exists an isometry $\mathsf{E}^{(\text{had})}$ acting on the prover registers $SA$ plus an auxiliary register, which returns a triple $(u,v,w)$ satisfying
\begin{equation}\label{eq:hadamard}
    \theta_0 \oplus \theta_1 = u \oplus \langle v, s_0 \oplus s_1 \rangle \oplus \langle w, f(x_0, y) \oplus f(x_1, y) \rangle
\end{equation}
with probability $1 - O({\delta})$. Let 
\[
    \Pi^{(\text{sub},0,0)}_{VE} := \sum_{s_0,s_1,x_0,x_1} \ketbra{s_0,s_1,x_0,x_1}{s_0,s_1,x_0,x_1}_V \otimes
    \sum_{(u,v,w)\in S_{s_0,s_1,x_0,x_1,0,0}}
    \ketbra{u,v,w}{u,v,w}_E
\]
be the projection onto 
triples $(u,v,w)$ satisfying \cref{eq:hadamard}, where the variables $\theta_0 = \theta_1 = 0$ are fixed, which is classically controlled on the verifier register $V$. Define
\[
\Pi^{(\text{acc})} :=  (\mathsf{Id} \otimes \mathsf{E}^{(\text{had})})^\dagger \Pi^{(\text{sub},0,0)} (\mathsf{Id} \otimes \mathsf{E}^{(\text{had})})
\]
and $\Pi^{(\text{rej})} := \mathsf{Id} - \Pi^{(\text{acc})}$. Let $\mathsf{F}$ be the operator
\[
\mathsf{F} := \Pi^{(\text{acc})} -\Pi^{(\text{rej})} = 2\Pi^{(\text{acc})} - \mathsf{Id} 
\]
which is efficiently implementable since since both of the above projections are also efficient. We claim that 
\begin{equation}\label{eq:F_equality}
    \mathsf{F}\ket{\varphi_{0, 0}} \approx_{O(\sqrt{\delta})}   \ket{\varphi_{1, 0}}
\end{equation}
and we defer the proof of \cref{eq:F_equality} to a later point in this proof.

The following claim uses the (approximate) equalities proven above and it establishes the existence of an efficient simulator that, starting from the state of an \emph{honest} prover, produces a state close to the adversarial one.
\begin{claim}[Efficient Simulator]\label{lmm:simulator}
    Let
    \ifllncs
    \begin{align*}
        \ket{\psi_{\text{tar}}}_{RVS} = \frac{1}{\sqrt{2^{2n+2\lambda+3}}}\sum_{b,x_0,x_1, s_0, s_1, \theta_0, \theta_1} &\ket{x_0, x_1, s_0, s_1, \theta_0, \theta_1}_R\ket{x_0, x_1, s_0, s_1, \theta_0, \theta_1}_V \\&(-1)^{\theta_b}\ket{b,s_b, f(x_b, y)}_{S}
    \end{align*}
    \else\[
    \ket{\psi_{\text{tar}}}_{RVS} = \frac{1}{\sqrt{2^{2n+2\lambda+3}}}\sum_{b,x_0,x_1, s_0, s_1, \theta_0, \theta_1} \ket{x_0, x_1, s_0, s_1, \theta_0, \theta_1}_R\ket{x_0, x_1, s_0, s_1, \theta_0, \theta_1}_V (-1)^{\theta_b}\ket{b,s_b, f(x_b, y)}_{S}
    \]\fi
    let $\ket{\varphi_{\text{adv}}}$ as in \cref{eq:state} and let $\ket{\alpha}_A$ be the auxiliary state of the attacker. Then there exists a QPT simulator $\mathsf{Sim}$ acting on $SA$ such that
    \[
    \Tr_R \ket{\varphi_{\text{adv}}}_{RVSAE} \approx_{O(\sqrt{\delta})} \Tr_R \mathsf{Sim}(\ket{\psi_{\text{tar}}}_{RVS}\otimes \ket{\alpha}_A).
    \]
\end{claim}
In words, \cref{lmm:simulator} allows us to simulate the output of the adversary by activating the simulator on input $\ket{\psi_{\text{tar}}}$ and $\ket{\alpha}$. Let $\ket{\psi_{\text{tar}}} = \frac{1}{\sqrt{2}}(\ket{\psi_{\text{tar},0}} + \ket{\psi_{\text{tar},1}})$ where
\begin{align}\label{eq:target_state}
\ket{\psi_{\text{tar},b}} = &\frac{1}{\sqrt{2^{n+\lambda+1}}}\sum_{x_{1-b}, s_{1-b}, \theta_{1-b}} \ket{x_{1-b}, s_{1-b}, \theta_{1-b}}_R \ket{x_{1-b}, s_{1-b}, \theta_{1-b}}_V \nonumber\\ &\otimes 
\frac{1}{\sqrt{2^{n+\lambda+1}}}\sum_{x_b, s_b,  \theta_b} \ket{x_b, s_b, \theta_b}_R\ket{x_b, s_b, \theta_b}_V (-1)^{\theta_b}\ket{b,s_b, f(x_b, y)}_{S}.
\end{align}
By the \emph{distinguishing implies swapping} principle (see \cref{eq:distswap}), the state $\Tr_{RV}\ket{\psi_{\text{tar}}}$ is negligibly close (in trace distance) from a uniform mix over $\Tr_{RV}\ket{\psi_{\text{tar},b}}$. Therefore the output of the simulator can be approximated by activating $\mathsf{Sim}$ on the $S$ register of $\ket{\psi_{\text{tar},b}}$, for a randomly sampled $b\in\{0,1\}$, and $\ket{\alpha}$. The former can be rewritten as
\[
\Tr_{RV}\frac{1}{2} \sum_b \ketbra{\psi_{\text{tar},b}}{\psi_{\text{tar},b}}_{RVS} = \frac{1}{2^{1+\lambda+n}} \sum_{b,s,x} \ketbra{b}{b} \otimes \ketbra{s}{s} \otimes \ketbra{f(x, y)}{f(x, y)}
\]
which is easy to simulate given $f(x,y)$, for a randomly sampled $x$, from the ideal functionality. This concludes our proof.
\end{proof}
We conclude this section by proving the remaining technical claims.

\begin{proof}[Proof of \cref{eq:uhlmann}]
    Since the adversary succeeds with probability $1-O(\delta)$ and one of the conditions is that $d \neq 0^\lambda$, \ifllncs\else by \cref{lmm:gentle} \fi it suffices to prove the following identity
    \begin{align*}
\Tr_R &
\sum_{b,c,d\neq 0^\lambda,z_0,z_1} \ket{z_0,z_1, c,d,\theta_0, \theta_1}_R \ket{x_0,x_1, s_0,s_1,\theta_0, \theta_1}_V\ket{\varphi_{c,d,z_b, b}}_{SA}
\nonumber \\
=
\Tr_R& 
\sum_b \sum_{x_{1-b}, s_{1-b}, \theta_{1-b}}\ket{{x_{1-b}, s_{1-b}, \theta_{1-b}}}_R \ket{x_{1-b}, s_{1-b}, \theta_{1-b}}_V  \\ &\otimes 
\sum_{z_b, c, d\neq 0^\lambda,\theta_b} \ket{\theta_b, c, d,  \xi_{\theta_b, c, d, z_b}}_R \ket{x_b, s_b, \theta_b}_V \ket{\varphi_{c,d,z_b, b}}_{SA}.
\end{align*}
It suffices to show the existence of a unitary, acting only on the purification register $R$, mapping the state on the LHS to the state on the RHS.

Let $M_d$ be any $\lambda \times \lambda$ invertible matrix whose last row is $d$, which always exists if $d \neq 0^\lambda$. Note that $M_d \cdot x = y$ with $y_\lambda = \langle x, d \rangle$ and furthermore, since $M_d$ is invertible and linear, the set of pre-images $x$ with $M_d \cdot x = (\cdots, b)$ is exactly the set of solutions to the equation $\langle x, d\rangle = b$.

Define $\mathsf{U}({c,d,x_b, \theta_b})$ to be the unitary that implements the following steps:
\begin{align*}
    \ket*{0^{\lambda}} \xrightarrow{H^{\otimes (\lambda-1)}\otimes \mathsf{Id}}& \frac{1}{\sqrt{2^{\lambda-1}}}\sum_{u\in\{0,1\}^{\lambda-1}} \ket{u}\ket{0} \\
    \xrightarrow{\mathsf{Id}\otimes \mathsf{X}^{\langle x_b,c \rangle \oplus \theta_b}}& \frac{1}{\sqrt{2^{\lambda-1}}}\sum_{u\in\{0,1\}^{\lambda-1}} \ket{u}\ket{\langle x_b,c \rangle \oplus \theta_b} \\
    \xrightarrow{M_d^{-1}}& \frac{1}{\sqrt{2^{\lambda-1}}}\sum_{v:\langle v, d\rangle =\langle x_b,c \rangle \oplus \theta_b} \ket{v}.
\end{align*}
Let $\mathsf{U}^{(b)}$ be the controlled application of the above unitary on classical variables, that is
\[
\mathsf{U}^{(b)} = \sum_{c,d\neq 0^\lambda,x_b,\theta_b} \ketbra{c,d,x_b,\theta_b}{c,d,x_b,\theta_b}\otimes \mathsf{U}({c,d,x_b, \theta_b})^\dagger
\]
where $\mathsf{U}({c,d,x_b, \theta_b})^\dagger$ acts on the register storing the variable $r_b$. Note that $\mathsf{U}^{(0)}\mathsf{U}^{(1)} = \mathsf{U}^{(1)}\mathsf{U}^{(0)}$ since they are only intersecting on registers where they are classically controlled.

Then for all $b\in\{0,1\}$ we have that

\ifllncs
\begin{align*}
    &\mathsf{U}^{(b)}\mathsf{U}^{(1-b)}\sum_{c,d\neq 0^\lambda,z_0,z_1} \ket{z_0,z_1, c,d,\theta_0, \theta_1}_R \ket{x_0,x_1, s_0,s_1,\theta_0, \theta_1}_V\ket{\varphi_{c,d,z_b, b}}_{SA}\\
    =&\mathsf{U}^{(b)}\mathsf{U}^{(1-b)}
    \sum_{c,d\neq 0^\lambda}\ket{c,d}_R\sum_{z_{1-b}}\ket{z_{1-b}, \theta_{1-b}}_R\ket{x_{1-b}, s_{1-b}, \theta_{1-b}}_V
    \sum_{z_b} \ket{z_b, \theta_b}_R \ket{x_b,s_b,\theta_b}_V\ket{\varphi_{c,d,z_b, b}}_{SA}\\
=&\mathsf{U}^{(b)}
    \sum_{c,d\neq 0^\lambda}\ket{c,d}_R\sum_{x_{1-b}, s_{1-b}, \theta_{1-b}}\ket{x_{1-b},s_{1-b}, \theta_{1-b}}_R\ket{x_{1-b}, s_{1-b}, \theta_{1-b}}_V \\ 
    &\ \ \ \ \ \ \ \ \ \sum_{z_b} \ket{z_b, \theta_b}_R \ket{x_b,s_b,\theta_b}_V\ket{\varphi_{c,d,z_b, b}}_{SA}\\
=&\sum_{x_{1-b}, s_{1-b}, \theta_{1-b}}\ket{{x_{1-b}, s_{1-b}, \theta_{1-b}}}_R \ket{x_{1-b}, s_{1-b}, \theta_{1-b}}_V \\
& \ \ \ \ \ \ \ \ \ \sum_{z_b, c, d\neq 0^\lambda,\theta_b} \ket{\theta_b, c, d,  \xi_{\theta_b, c, d, z_b}}_R \ket{x_b, s_b, \theta_b}_V \ket{\varphi_{c,d,z_b, b}}_{SA}.
\end{align*}
\else
\begin{align*}
    &\mathsf{U}^{(b)}\mathsf{U}^{(1-b)}\sum_{c,d\neq 0^\lambda,z_0,z_1} \ket{z_0,z_1, c,d,\theta_0, \theta_1}_R \ket{x_0,x_1, s_0,s_1,\theta_0, \theta_1}_V\ket{\varphi_{c,d,z_b, b}}_{SA}\\
    =&\mathsf{U}^{(b)}\mathsf{U}^{(1-b)}
    \sum_{c,d\neq 0^\lambda}\ket{c,d}_R\sum_{z_{1-b}}\ket{z_{1-b}, \theta_{1-b}}_R\ket{x_{1-b}, s_{1-b}, \theta_{1-b}}_V
    \sum_{z_b} \ket{z_b, \theta_b}_R \ket{x_b,s_b,\theta_b}_V\ket{\varphi_{c,d,z_b, b}}_{SA}\\
=&\mathsf{U}^{(b)}
    \sum_{c,d\neq 0^\lambda}\ket{c,d}_R\sum_{x_{1-b}, s_{1-b}, \theta_{1-b}}\ket{x_{1-b},s_{1-b}, \theta_{1-b}}_R\ket{x_{1-b}, s_{1-b}, \theta_{1-b}}_V
    \sum_{z_b} \ket{z_b, \theta_b}_R \ket{x_b,s_b,\theta_b}_V\ket{\varphi_{c,d,z_b, b}}_{SA}\\
=&\sum_{x_{1-b}, s_{1-b}, \theta_{1-b}}\ket{{x_{1-b}, s_{1-b}, \theta_{1-b}}}_R \ket{x_{1-b}, s_{1-b}, \theta_{1-b}}_V 
\sum_{z_b, c, d\neq 0^\lambda,\theta_b} \ket{\theta_b, c, d,  \xi_{\theta_b, c, d, z_b}}_R \ket{x_b, s_b, \theta_b}_V \ket{\varphi_{c,d,z_b, b}}_{SA}.
\end{align*}
\fi
where the first identity is a re-arrangement of the registers, the second one holds because $\mathsf{U}^{(1-b)}$ precisely uncomputes $r_{1-b}$ and, by linearity, there are exactly $2^{\lambda-1}$ values of $r_{1-b}$ consistent with $\theta_{1-b} = 0$ and 
$\theta_{1-b} = 1$, so the amplitudes of $\theta_{1-b}$ are uniform. Finally, the last equality follows by setting
\[
\mathsf{U}^{(b)}\ket{c,d,\theta_b,x_b,s_b, r_b}_R = \ket{c,d,\theta_b}_R \underbrace{\ket{x_b,s_b}_R \mathsf{U}({c,d,x_b, \theta_b})^\dagger \ket{r_b}_R}_{=: \ket{\xi_{\theta_b, c, d, z_b}}}
\]
and by a re-arrangement of the registers.
\end{proof}

\begin{proof}[Proof of \cref{eq:phi_equality} and \cref{eq:F_equality}]
    Let $\Pi^{(\text{sub})}_{VE}$ and $\mathsf{E}^{(\text{had})}$ defined as in \cref{claim:had_extractor} and let us rewrite the state as
    \[
    \frac{1}{\sqrt{2}}\sum_b  \sum_{\theta_{1-b}}\ket{\theta_{1-b}}_R\ket{\theta_{1-b}}_V \otimes  \sum_{\theta_b} \ket{\theta_b}_R \ket{\theta_b}_V \otimes \ket{\varphi_{b, \theta_b}}_{RVSAE} = \frac{1}{\sqrt{2}} \sum_b \ket{\varphi_b}_{RSVAE}
    \]
    where we have re-introduced the normalization factor for the $b$ variable. By \cref{claim:had_extractor} we have that
    \begin{equation}\label{eq:normsqared}        
    \norm{\frac{1}{\sqrt{2}} \sum_b  \Pi^{(\text{sub})}_{VE}(\mathsf{Id} \otimes \mathsf{E}^{(\text{had})})\ket{\varphi_b}_{RSVAE}}^2 = 1 - O({\delta})
    \end{equation}
    At the same time, for all $b\in\{0,1\}$ we can expand the above state as
    \[
   \Pi^{(\text{sub})}_{VE}(\mathsf{Id} \otimes \mathsf{E}^{(\text{had})}) \ket{\varphi_b}_{RSVAE} = \Pi^{(\text{sub})}_{VE}\sum_{\theta_{1-b}}\ket{\theta_{1-b}}_R\ket{\theta_{1-b}}_V \otimes  \sum_{\theta_b} \ket{\theta_b}_R \ket{\theta_b}_V \otimes (\mathsf{Id} \otimes \mathsf{E}^{(\text{had})})\ket{\varphi_{b, \theta_b}}_{RVSAE}
    \]
    and note that the register containing $\theta_{1-b}$ is in tensor product with the rest of the system, so the probability that the projection $\Pi^{(\text{sub})}$ succeeds is precisely $1/2$, since it is testing a linear relation. We can conclude that
    \[
    \norm{\Pi^{(\text{sub})}_{VE}(\mathsf{Id} \otimes \mathsf{E}^{(\text{had})}) \ket{\varphi_b}_{RSVAE}}^2 = \frac{1}{2}.
    \]
    Let us define $\ket{\tilde{\varphi}_b}:=\Pi^{(\text{sub})}_{VE}(\mathsf{Id} \otimes \mathsf{E}^{(\text{had})}) \ket{\varphi_b}$, then by \cref{eq:normsqared} we have that
    \[
    \norm{\sum_b\ket{\tilde{\varphi}_b}}^2= \sum_b\norm{\ket{\tilde{\varphi}_b}}^2 + 2\Re\braket{\tilde{\varphi}_0}{\tilde{\varphi}_1} =2(1-O({\delta}))
    \]
    thus the real component of $\braket{\tilde{\varphi}_0}{\tilde{\varphi}_1}$ is bounded by $1/2 - O({\delta})$ and furthermore
    \[
    \abs{1/2 - O({\delta})}\leq \abs{\braket{\tilde{\varphi}_0}{\tilde{\varphi}_1}} \leq \norm{\ket{\tilde{\varphi}_0}} \cdot \norm{\ket{\tilde{\varphi}_1}} = \frac{1}{2}
    \]
    by the Cauchy-Schwarz inequality. Finally, we obtain that the trace distance between $\ket{\tilde{\varphi}_0}$ and $\ket{\tilde{\varphi}_1}$ is bounded by
    \[
    \frac{1}{2}\sqrt{1-4\abs{\braket{\tilde{\varphi}_0}{\tilde{\varphi}_1}}^2} \leq 
    \frac{1}{2}\sqrt{1-4\abs{1/2 - O({\delta})}^2} \leq O(\sqrt{\delta})
    \]
    since the function is decreasing in $\abs{\braket{\tilde{\varphi}_0}{\tilde{\varphi}_1}}$. To summarize, we obtained that 
    \begin{equation}\label{eq:accept}
        \Pi^{(\text{sub})}_{VE}(\mathsf{Id} \otimes \mathsf{E}^{(\text{had})}) \ket{\varphi_0} \approx_{O(\sqrt{\delta})} \Pi^{(\text{sub})}_{VE}(\mathsf{Id} \otimes \mathsf{E}^{(\text{had})}) \ket{\varphi_1}
    \end{equation}
    and, using a symmetric argument, i.e., the fact that the rejection probability is bounded by $O({\delta})$, one can also prove that
\begin{equation}\label{eq:reject}
        (\mathsf{Id}-\Pi^{(\text{sub})}_{VE})(\mathsf{Id} \otimes \mathsf{E}^{(\text{had})}) \ket{\varphi_0} \approx_{O(\sqrt{\delta})} -(\mathsf{Id}-\Pi^{(\text{sub})}_{VE})(\mathsf{Id} \otimes \mathsf{E}^{(\text{had})}) \ket{\varphi_1}.
    \end{equation}   
For notational convenience, let us rewrite 
\[
\Pi^{(\text{sub})}_{VE}=\sum_{\theta_0, \theta_1} \ketbra{\theta_0, \theta_1}{ \theta_0, \theta_1}_V \otimes
    \sum_{t: P(t)= \theta_0 \oplus\theta_1}
    \ketbra{t}{t}_{VE}
\]
where $t = (s_0,s_1,x_0,x_1,u,v,w)$ and $P(t)$ as the predicate defined by \cref{eq:hadamard}. Consequently, we also have
\[
\mathsf{Id}- \Pi^{(\text{sub})}_{VE}=\sum_{\theta_0, \theta_1} \ketbra{\theta_0, \theta_1}{ \theta_0, \theta_1}_V \otimes
    \sum_{t: P(t)\neq \theta_0 \oplus\theta_1}
    \ketbra{t}{t}_{VE}.
\]
For any fixed pair $(\theta_0, \theta_1)$, the action of the operators $\Pi^{(\text{sub})}$ and $(\mathsf{Id}- \Pi^{(\text{sub})})$ splits the state $\ket{\varphi_b}$ into two orthogonal subspaces. Furthermore, since orthogonal projections do not increase vector lengths, for any fixed pair $(\theta_0, \theta_1)$ we also have
\begin{align*}
&\Pi^{(\text{sub})}_{VE}\ket{\theta_0, \theta_1}_R\ket{\theta_0, \theta_1}_V \mathsf{E}^{(\text{had})}\ket{\varphi_{0, \theta_0}}_{RVSAE}\\ 
    &=\ket{\theta_0, \theta_1}_R\ket{\theta_0, \theta_1}_V \left(\sum_{t: P(t)= \theta_0 \oplus\theta_1}
    \ketbra{t}{t}_{VE}\right)\mathsf{E}^{(\text{had})}\ket{\varphi_{0, \theta_0}}_{RVSAE}\\ &\approx_{O(\sqrt{\delta})}
    \ket{\theta_0, \theta_1}_R\ket{\theta_0, \theta_1}_V \left(\sum_{t: P(t)= \theta_0 \oplus\theta_1}
    \ketbra{t}{t}_{VE}\right)\mathsf{E}^{(\text{had})}\ket{\varphi_{1, \theta_1}}_{RVSAE} 
\end{align*}
by \cref{eq:accept}, and similarly
\begin{align*}
&\left(\mathsf{Id}-\Pi^{(\text{sub})}_{VE}\right)\ket{\theta_0, \theta_1}_R\ket{\theta_0, \theta_1}_V \mathsf{E}^{(\text{had})}\ket{\varphi_{0, \theta_0}}_{RVSAE}\\ 
    &=\ket{\theta_0, \theta_1}_R\ket{\theta_0, \theta_1}_V \left(\sum_{t: P(t)\neq \theta_0 \oplus\theta_1}
    \ketbra{t}{t}_{VE}\right)\mathsf{E}^{(\text{had})}\ket{\varphi_{0, \theta_0}}_{RVSAE}\\ &\approx_{O(\sqrt{\delta})}
    \ket{\theta_0, \theta_1}_R\ket{\theta_0, \theta_1}_V (-1)\left(\sum_{t: P(t)\neq \theta_0 \oplus\theta_1}
    \ketbra{t}{t}_{VE}\right)\mathsf{E}^{(\text{had})}\ket{\varphi_{1, \theta_1}}_{RVSAE} 
\end{align*}
by \cref{eq:reject}. The validity of \cref{eq:F_equality} then follows by combining the above two identities with $(\theta_0, \theta_1) = (0,0)$ and a triangle inequality. Furthermore,  fixing $(\theta_0, \theta_1) = (0,0)$ first and then $(\theta_0, \theta_1) = (1,0)$, we can derive
\begin{align*}
    \left(\sum_{t: P(t)= 0}
    \ketbra{t}{t}_{VE}\right)\mathsf{E}^{(\text{had})}\ket{\varphi_{0, 0}}_{RVSAE} \approx_{O(\sqrt{\delta})} 
&\left(\sum_{t: P(t)= 0}
    \ketbra{t}{t}_{VE}\right)\mathsf{E}^{(\text{had})}\ket{\varphi_{1, 0}}_{RVSAE}\\
   \approx_{O(\sqrt{\delta})} 
& (-1)  \left(\sum_{t: P(t)= 0}
    \ketbra{t}{t}_{VE}\right)\mathsf{E}^{(\text{had})}\ket{\varphi_{0, 1}}_{RVSAE}.
\end{align*}
Similarly, we can derive 
\begin{align*}
    \left(\sum_{t: P(t)= 1}
    \ketbra{t}{t}_{VE}\right)\mathsf{E}^{(\text{had})}\ket{\varphi_{0, 0}}_{RVSAE} \approx_{O(\sqrt{\delta})} 
&(-1)\left(\sum_{t: P(t)= 1}
    \ketbra{t}{t}_{VE}\right)\mathsf{E}^{(\text{had})}\ket{\varphi_{1, 0}}_{RVSAE}\\
\approx_{O(\sqrt{\delta})} 
& (-1)  \left(\sum_{t: P(t)= 1}
   \ketbra{t}{t}_{VE}\right)\mathsf{E}^{(\text{had})}\ket{\varphi_{0, 1}}_{RVSAE}.
\end{align*}
Then, by a triangle inequality and by unitary invariance of the trace distance,  we obtain
\[
\ket{\varphi_{0, 0}}_{RVSAE} \approx_{O(\sqrt{\delta})} - \ket{\varphi_{0, 1}}_{RVSAE}
\]
and the other approximate equality in \cref{eq:phi_equality} can be shown with the same strategy.

\end{proof}

\begin{proof}[Proof of \cref{lmm:simulator}]
    We describe the simulator. On input $\ket{\alpha}$, the simulator runs the first round interaction (including the computation of the verifier), followed by the extractor $\mathsf{E}^{(\text{comp})}$ to compute a state that, by \cref{eq:comp_basis_ind}, is within $O(\sqrt{\delta})$ trace distance from $\ket{\varphi_{\text{adv}}}$. 
    
    To avoid ambiguity, we relabel the registers for the state $\ket{\varphi_{\text{tar}}}$ as $\Tilde{R}\Tilde{V}\Tilde{S}$. Acting on the $\Tilde{S}$ register of $\ket{\varphi_{\text{tar}}}$, the simulator measures the first qubit of $E$ and the first qubit of $\Tilde{S}$ with the projective measurement
    \[
    \Pi^{(0)} = \sum_{(a,b) : a\oplus b = 0} \ketbra{a,b}{a,b} \quad\text{and}\quad
    \Pi^{(1)} = \mathsf{Id} -  \Pi^{(0)}.
    \]
    If the result of the measurement is $1$, the simulator proceeds as follows:
    \begin{itemize}
        \item Controlled on the first qubit of $\Tilde{S}$ being $0$, apply the following operators:
        \begin{itemize}
            \item Controlled on the value $\theta_1 =1$ in register $V$, apply a $\mathsf{Z}$ (phase flip) gate to the first qubit of $E$.
            \item Let $\Tilde{\mathsf{F}}$ be defined as $\mathsf{F}$, except that it is classically controlled on the $(s_1, x_1)$ variables in register $V$ and the $(s_0, f(x_0,y))$ variables in register $\Tilde{S}$. Apply $\Tilde{\mathsf{F}}^\dagger$. %
            \item $\mathsf{CNOT}$ the content of the register $E$ onto the register $\Tilde{S}$.
        \end{itemize}
        \item Controlled on the first qubit of $\Tilde{S}$ being $1$, 
        apply the following operators:
        \begin{itemize}
            \item Controlled on the value $\theta_0 =1$ in register $V$, apply a $\mathsf{Z}$ (phase flip) gate to the first qubit of $E$.
            \item Apply the $\Tilde{\mathsf{F}}$ operator, classically controlled on the $(s_0, x_0)$ variables in register $V$ and the $(s_1, f(x_1,y))$ variables in register $\Tilde{S}$.
            \item $\mathsf{CNOT}$ the content of the register $E$ onto the register $\Tilde{S}$.
        \end{itemize}
    \end{itemize}
    Finally, the simulator traces out the registers $R$, $V$, and $\Tilde{S}$.
    On the other hand, if the result of the measurement is $0$, then we proceed as above except that, after the application of the phase flip, we add an application of the $\mathsf{F}$ operator \emph{only} controlled on the $V$ register (in both cases). This is \emph{in addition} to the above instructions, which are also executed by the simulator.

    It is clear that the simulator runs in QPT and only acts on $\ket{\alpha}$ and the $\Tilde{S}$ register, plus some internal registers. All is left to be shown is that the state produced by the simulator is approximately equal to $\ket{\varphi_{\text{adv}}}_{\tilde{R}\tilde{V}SAE}$. The first projective measurement has the effect of entangling the internal state of the simulator with the honest target state. More precisely, the joint state of the system after observing the outcome $1$ is $O(\sqrt{\delta})$-close to 
    \begin{align*}
        \sum_{c} \ket{\varphi_{\text{tar}, c}}_{\Tilde{R}\Tilde{V}\Tilde{S}} \otimes \sum_{\theta_{c}}\ket{\theta_{c}}_R\ket{\theta_{c}}_V \otimes  \sum_{\theta_{1-c}} \ket{\theta_{1-c}}_R \ket{\theta_{1-c}}_V \otimes \ket{\varphi_{{1-c}, \theta_{1-c}}}_{RVSAE}
    \end{align*}
    omitting normalization factors, where $\ket{\varphi_{\text{tar}, c}}$ is defined as in \cref{eq:target_state} and the other terms in the tensor product are derived from \cref{eq:state}. Next, we analyze the effect of the operations controlled on $c=0$ (the case $c=1$ is identical but with the variables swapped). The effect of the controlled phase flip is to turn the state into
    \begin{align}\label{eq:intermediate_state}
        &\ket{\varphi_{\text{tar}, 0}}_{\Tilde{R}\Tilde{V}\Tilde{S}} \otimes \sum_{\theta_{0}}\ket{\theta_{0}}_R\ket{\theta_{0}}_V \otimes  \sum_{\theta_{1}} \ket{\theta_{1}}_R \ket{\theta_{1}}_V \otimes (-1)^{\theta_1}\ket{\varphi_{{1}, \theta_{1}}}_{RVSAE}\nonumber\\
        \approx_{O(\sqrt{\delta})}&\ket{\varphi_{\text{tar}, 0}}_{\Tilde{R}\Tilde{V}\Tilde{S}} \otimes \sum_{\theta_{0}}\ket{\theta_{0}}_R\ket{\theta_{0}}_V \otimes  \sum_{\theta_{1}} \ket{\theta_{1}}_R \ket{\theta_{1}}_V \otimes \ket{\varphi_{{1}, 0}}_{RVSAE}
    \end{align}
    by \cref{eq:phi_equality}. Expanding the above state we obtain
    \begin{align*}
        &\sum_{x_{1}, s_{1}, \theta_{1}} \ket{x_{1}, s_{1}, \theta_{1}}_{\tilde{R}} \ket{x_{1}, s_{1}, \theta_{1}}_{\tilde{V}} \nonumber\\ &\otimes 
\sum_{x_0, s_0,  \theta_0} \ket{x_0, s_0, \theta_0}_{\tilde{R}}\ket{x_0, s_0, \theta_0}_{\tilde{V}} (-1)^{\theta_0}\ket{0,s_0, f(x_0, y)}_{\tilde{S}}\\
&\otimes \sum_{\theta_{0}}\ket{\theta_{0}}_R\ket{\theta_{0}}_V \otimes  \sum_{\theta_{1}} \ket{\theta_{1}}_R \ket{\theta_{1}}_V \\
&\otimes \sum_{x_{0}, s_{0}}\ket{{x_{0}, s_{0}}}_R \ket{x_{0}, s_{0}}_V \otimes \sum_{z_1, c, d} \ket{c, d,  \xi_{0, c, d, z_1}}_R \Pi^{(\text{eq},1)}\ket{x_1, s_1}_V \mathsf{E}^{(\text{comp})}\ket{\varphi_{c,d,z_1, 1}}_{SA}.
    \end{align*}
    For the analysis, we can equivalently substitute the application of $\Tilde{\mathsf{F}}^\dagger$ with the application of $\mathsf{F}^\dagger$ controlled on the variables $(s_0, x_0)$ in register $\Tilde{V}$ (instead of $\Tilde{S}$), which does not change the resulting state. Let us define the operator $\mathsf{U}^{(f,y)}$ acting on registers $\tilde{V}$ and $\Tilde{S}$, as
    \[
    \mathsf{U}^{(f,y)} \ket{s,x} \ket{b,w,z} = \ket{s,x}\ket{b,w\oplus s, z \oplus f(x,y)} 
    \]
    in words, it uncomputes the content of $\Tilde{S}$. Since $\mathsf{F}$ is classically controlled on $\tilde{V}$ and $\mathsf{U}^{(f,y)}$ is diagonal in the computational basis, they commute and we have that
    \[
    \mathsf{F}^\dagger = \mathsf{Id}\cdot \mathsf{F}^\dagger = (\mathsf{U}^{(f,y)})^\dagger\mathsf{U}^{(f,y)}\mathsf{F}^\dagger = (\mathsf{U}^{(f,y)})^\dagger\mathsf{F}^\dagger\mathsf{U}^{(f,y)}.
    \]
    Next, we analyze the application of the three operators to the above state. The application of $\mathsf{U}^{(f,y)}$ results into the state
\begin{align*}
        &\sum_{x_{1}, s_{1}, \theta_{1}} \ket{x_{1}, s_{1}, \theta_{1}}_{\tilde{R}} \ket{x_{1}, s_{1}, \theta_{1}}_{\tilde{V}} \nonumber \otimes 
\sum_{x_0, s_0,  \theta_0} \ket{x_0, s_0, \theta_0}_{\tilde{R}}\ket{x_0, s_0, \theta_0}_{\tilde{V}} (-1)^{\theta_0}\\
&\otimes \sum_{\theta_{0}}\ket{\theta_{0}}_R\ket{\theta_{0}}_V \otimes  \sum_{\theta_{1}} \ket{\theta_{1}}_R \ket{\theta_{1}}_V \\
&\otimes \sum_{x_{0}, s_{0}}\ket{{x_{0}, s_{0}}}_R \ket{x_{0}, s_{0}}_V \otimes \sum_{z_1, c, d} \ket{c, d,  \xi_{0, c, d, z_1}}_R \Pi^{(\text{eq}, 1)}\ket{x_1, s_1}_V \mathsf{E}^{(\text{comp})}\ket{\varphi_{c,d,z_1, 1}}_{SA}\\
=&\sum_{x_{1}, s_{1}} \ket{x_{1}, s_{1}}_{\tilde{R}} \ket{x_{1}, s_{1}}_{\tilde{V}} \otimes \sum_{\theta_1}\ket{\theta_1}_{\Tilde{R}}\ket{\theta_1}_{\Tilde{V}}\nonumber
\otimes \sum_{\theta_0}\ket{\theta_0}_{\Tilde{R}}\ket{\theta_0}_{\Tilde{V}}(-1)^{\theta_0}
\\
&\otimes \sum_{\theta_{0}}\ket{\theta_{0}}_R\ket{\theta_{0}}_V \otimes  \sum_{\theta_{1}} \ket{\theta_{1}}_R \ket{\theta_{1}}_V \otimes \sum_{x_{0}, s_{0}}\ket{{x_{0}, s_{0}}}_R \ket{x_{0}, s_{0}}_V \\
&\otimes 
\underbrace{\sum_{x_0, s_0} \ket{x_0, s_0}_{\tilde{R}}\ket{x_0, s_0}_{\tilde{V}} 
\otimes \sum_{z_1, c, d} \ket{c, d,  \xi_{0, c, d, z_1}}_R \Pi^{(\text{eq},1)} \ket{x_1, s_1}_V \mathsf{E}^{(\text{comp})}\ket{\varphi_{c,d,z_1, 1}}_{SA}}_{= \ket{\varphi_{1,0}}}
    \end{align*}
temporarily tracing out $\tilde{S}$, since it is now in the $\ket{0,0,0}$ state, and re-arranging the registers. Applying the $\mathsf{F}$ operator and again re-arranging registers, results into
\begin{align*}
    &\sum_{x_{1}, s_{1}} \ket{x_{1}, s_{1}}_{\tilde{R}} \ket{x_{1}, s_{1}}_{\tilde{V}} 
\otimes \sum_{\theta_{0}}\ket{\theta_{0}}_R\ket{\theta_{0}}_V \otimes  \sum_{\theta_{1}} \ket{\theta_{1}}_R \ket{\theta_{1}}_V \otimes \sum_{x_{0}, s_{0}}\ket{{x_{0}, s_{0}}}_R \ket{x_{0}, s_{0}}_V \\
&\otimes \sum_{\theta_1}\ket{\theta_1}_{\Tilde{R}}\ket{\theta_1}_{\Tilde{V}}
\otimes \sum_{\theta_0}\ket{\theta_0}_{\Tilde{R}}\ket{\theta_0}_{\Tilde{V}}
\otimes(-1)^{\theta_0}\mathsf{F}\ket{\varphi_{1,0}}_{\Tilde{R}\Tilde{V}RVSAE}\\
\approx_{O(\sqrt{\delta})}
&\sum_{x_{1}, s_{1}} \ket{x_{1}, s_{1}}_{\tilde{R}} \ket{x_{1}, s_{1}}_{\tilde{V}} 
\otimes \sum_{\theta_{0}}\ket{\theta_{0}}_R\ket{\theta_{0}}_V \otimes  \sum_{\theta_{1}} \ket{\theta_{1}}_R \ket{\theta_{1}}_V \otimes \sum_{x_{0}, s_{0}}\ket{{x_{0}, s_{0}}}_R \ket{x_{0}, s_{0}}_V \\
&\otimes \sum_{\theta_1}\ket{\theta_1}_{\Tilde{R}}\ket{\theta_1}_{\Tilde{V}}
\otimes \sum_{\theta_0}\ket{\theta_0}_{\Tilde{R}}\ket{\theta_0}_{\Tilde{V}}
\otimes(-1)^{\theta_0}\ket{\varphi_{0,0}}_{\Tilde{R}\Tilde{V}RVSAE}\\
\approx_{O(\sqrt{\delta})}
&\sum_{x_{1}, s_{1}} \ket{x_{1}, s_{1}}_{\tilde{R}} \ket{x_{1}, s_{1}}_{\tilde{V}} 
\otimes \sum_{\theta_{0}}\ket{\theta_{0}}_R\ket{\theta_{0}}_V \otimes  \sum_{\theta_{1}} \ket{\theta_{1}}_R \ket{\theta_{1}}_V \otimes \sum_{x_{0}, s_{0}}\ket{{x_{0}, s_{0}}}_R \ket{x_{0}, s_{0}}_V \\
&\otimes \sum_{\theta_1}\ket{\theta_1}_{\Tilde{R}}\ket{\theta_1}_{\Tilde{V}}
\otimes \sum_{\theta_0}\ket{\theta_0}_{\Tilde{R}}\ket{\theta_0}_{\Tilde{V}}
\otimes\ket{\varphi_{0,\theta_0}}_{\Tilde{R}\Tilde{V}RVSAE}
\end{align*}
where the first identity follows by \cref{eq:F_equality} and the second by \cref{eq:phi_equality}. Since the subsequent operators only act on 
$\ket{\varphi_{0,\theta_0}}$, it suffices to analyze their action on the expanded state
\begin{align} \label{eq:finalfinal}
&\ket{\varphi_{0,\theta_0}}_{\Tilde{R}\Tilde{V}RVSAE} \nonumber\\
=& 
\sum_{x_{1}, s_{1}}\ket{{x_{1}, s_{1}}}_R \ket{x_{1}, s_{1}}_V \otimes \sum_{z_0, c, d} \ket{c, d,  \xi_{\theta_0, c, d, z_0}}_{\tilde{R}} \Pi^{(\text{eq},0)}\ket{x_0, s_0}_{\tilde{V}} \mathsf{E}^{(\text{comp})}\ket{\varphi_{c,d,z_0, 0}}_{SA}.
\end{align}
Here, we can simply observe that the above state is supported only on states where the $E$ register contains the triple $(0,s_0, f(x_0, y))$, by definition of the projector $\Pi^{(\text{eq},0)}$. Thus the application of  $(\mathsf{U}^{(f,y)})^\dagger$ is exactly canceled out by the $\mathsf{CNOT}$ on the same register, i.e., we have
\[
\mathsf{CNOT}_{E\Tilde{S}}(\mathsf{U}^{(f,y)})^\dagger_{\Tilde{V}\Tilde{S}}\ket{\varphi_{0,\theta_0}}_{\Tilde{R}\Tilde{V}RVSAE}\otimes \ket{0,0,0}_{\Tilde{S}} = \ket{\varphi_{0,\theta_0}}_{\Tilde{R}\Tilde{V}RVSAE}\otimes \ket{0,0,0}_{\Tilde{S}}.
\]
Finally, note that the registers $R$ and $V$ in the state in \cref{eq:finalfinal} are in tensor product with the rest of the state, so swapping them with the corresponding registers $\Tilde{R}$ and $\Tilde{V}$ does not change the state. Tracing out the registers $R$ and $V$, we obtain the state
\begin{align*}
& \sum_{\theta_1}\ket{\theta_1}_{\Tilde{R}}\ket{\theta_1}_{\Tilde{V}}
\otimes \sum_{\theta_0}\ket{\theta_0}_{\Tilde{R}}\ket{\theta_0}_{\Tilde{V}}
\\ &\otimes\sum_{x_{1}, s_{1}}\ket{{x_{1}, s_{1}}}_{\Tilde{R}} \ket{x_{1}, s_{1}}_{\Tilde{V}}\otimes \sum_{z_0, c, d} \ket{c, d,  \xi_{\theta_0, c, d, z_0}}_{\tilde{R}} \Pi^{(\text{eq}, 0)}\ket{x_0, s_0}_{\tilde{V}} \mathsf{E}^{(\text{comp})}\ket{\varphi_{c,d,z_0, 0}}_{SA}\\
    &= \sum_{\theta_1}\ket{\theta_1}_{\Tilde{R}}\ket{\theta_1}_{\Tilde{V}}
\otimes \sum_{\theta_0}\ket{\theta_0}_{\Tilde{R}}\ket{\theta_0}_{\Tilde{V}}
\otimes\ket{\varphi_{0,\theta_0}}_{\Tilde{R}\Tilde{V}SAE}.
\end{align*}
Combining with same analysis for the bit $c=1$, overall we obtain that the joint state is within $O(\sqrt{\delta})$ trace distance from
\[
\sum_c \sum_{\theta_{1-c}}\ket{\theta_{1-c}}_{\Tilde{R}}\ket{\theta_{1-c}}_{\Tilde{V}}
\otimes \sum_{\theta_c}\ket{\theta_c}_{\Tilde{R}}\ket{\theta_c}_{\Tilde{V}}
\otimes\ket{\varphi_{c,\theta_c}}_{\Tilde{R}\Tilde{V}SAE}.
\]
which is precisely the state in \cref{eq:state}.

This concludes the proof in case the measurement of the simulator returns $1$. On the other hand, if the measurement returns $0$ (again we only show the case controlled on $c=0$ and the case $c=1$ follows with the same argument plus the fact that $\mathsf{F}^\dagger = \mathsf{F}$) with the same argument as above we can see that the joint state after the phase flip equals
\[
\ket{\varphi_{\text{tar}, 0}}_{\Tilde{R}\Tilde{V}\Tilde{S}} \otimes \sum_{\theta_{0}}\ket{\theta_{0}}_R\ket{\theta_{0}}_V \otimes  \sum_{\theta_{1}} \ket{\theta_{1}}_R \ket{\theta_{1}}_V \otimes \ket{\varphi_{{0}, 0}}_{RVSAE}.
\]
The application of the extra $\mathsf{F}$ operator, this time controlled only on $V$, results into 
\begin{align*}
    &\ket{\varphi_{\text{tar}, 0}}_{\Tilde{R}\Tilde{V}\Tilde{S}} \otimes \sum_{\theta_{0}}\ket{\theta_{0}}_R\ket{\theta_{0}}_V \otimes  \sum_{\theta_{1}} \ket{\theta_{1}}_R \ket{\theta_{1}}_V \otimes \mathsf{F}\ket{\varphi_{{0}, 0}}_{RVSAE}\\
    \approx_{O(\sqrt{\delta})}&\ket{\varphi_{\text{tar}, 0}}_{\Tilde{R}\Tilde{V}\Tilde{S}} \otimes \sum_{\theta_{0}}\ket{\theta_{0}}_R\ket{\theta_{0}}_V \otimes  \sum_{\theta_{1}} \ket{\theta_{1}}_R \ket{\theta_{1}}_V \otimes \ket{\varphi_{{1}, 0}}_{RVSAE}
\end{align*}
by \cref{eq:F_equality}. However, note that this state is now identical to the one in \cref{eq:intermediate_state}, so precisely the same argument as above applies.
\end{proof}

\subsection{Compressing Interactive Protocols} 
\label{sec:compress_inter_protocols}

Before presenting our compiler, let us make some simplifications to the input protocol, which do not affect the generality of our result but make the presentation cleaner. First, we can assume that the verifier samples a random seed $s \gets \{0,1\}^\lambda$ at the beginning of the protocol, and afterwards it is completely \emph{deterministic}. Assuming one-way functions, this is without loss of generality since $s$ can be expanded into arbitrarily-many random coins, using a sufficiently expanding PRG. Then the \emph{next message} function of the verifier for the $i$-th round can be expressed as
\[
f_i(s, m_1, \dots, m_i)
\]
where $s$ is the aforementioned seed and $(m_1, \dots, m_i)$ are the \emph{prover} messages received so far. This is once again without loss of generality: Although in general the verifier's next message function could also depend on the previous messages of the verifier, these can be recomputed from $s$ and the prover's messages.
Finally, assuming that the prover sends the last message and that $r = r(\lambda)$ is the number of rounds of the protocol, we can also assume that the output of the verifier is computed by a function $f_r(s, m_1, \dots, m_r) \to z$ of the protocol transcript.

With the above structure in mind, let us describe our compiler, assuming the existence of a chosen-input committed SFS protocol, a collision-resistant commitment, and a state-preserving succinct AoK for NP. Let $(P, V)$ be the prover and verifier algorithm, respectively, let $\chi$ be the public input and let $\varepsilon$ be the input parameter. As the first message, the verifier samples a commitment function $\mathsf{Com}$ and sends it to the prover. Then, for each $i \in \{1, \dots, r-1\}$ the parties proceed as follows.
\begin{itemize}
    \item (Prover Message) The prover computes locally its next message of the protocol $m_i$ and sends $c_i := \mathsf{Com}(m_i)$ to the verifier. Then, the prover and the verifier engage in a state-preserving succinct AoK of knowledge, where the prover proves the knowledge of a pre-image of $c_i$.
    \item (Verifier Message) The prover and the verifier engage in a chosen-input committed SFS protocol for $f_i$ with parameter $O(\varepsilon^{2})$, where the verifier's input are $s$ and the commitments $(c_1, \dots, c_i)$\footnote{In a slight abuse of notation, we assume that the concatenation of several commitments is also a valid commitment.} and the prover's input are the messages $(m_1, \dots, m_i)$. At the end of the interaction, the prover receives $f_i(s, m_1, \dots, m_i)$.
\end{itemize}
The prover computes its last message $m_r$ and sends $c_r := \mathsf{Com}(m_r)$ to the verifier. Then, the prover and the verifier engage in a state-preserving succinct AoK of knowledge, where the prover proves the knowledge of a pre-image of $c_r$. Finally, the verifier sends $s$ to the prover. The prover returns $z$ and the prover and verifier then engage in a state-preserving succinct AoK, where the prover proves the knowledge of $(m_1, \dots, m_r)$ such that:
\begin{itemize}
    \item For all $i\in\{1, \dots, r\}$, it holds that $c_i:= \mathsf{Com}(m_i)$.
    \item $f_r(s,m_1, \dots, m_r) = z$.
\end{itemize}
The verifier outputs $z$ if all sub-protocols accept, and otherwise it outputs $\bot$.

The compiler clearly preserves the correctness of the protocol, and the efficiency follows from the properties of the succinct AoK and the chosen-input committed SFS. The following establishes the soundness of the compiler.

\begin{lemma}[Soundness]\label{lmm:soundnes}
    Assume the existence of succinct collision-resistant commitments, state-preserving succinct AoKs, and chosen-input committed SFS.
    The compiler as described above satisfies soundness according to \cref{def:ComSFS}.
\end{lemma}

\begin{proof}[Proof of \cref{lmm:soundnes}]
    Let $\tilde{\mathcal{A}}$ be a QPT adversary against the compiled protocol. We gradually change the view of the adversary in the following sequence of hybrids. For each round $i\in\{1, \dots, r\}$ we modify the experiment as follows:
    \begin{itemize}
        \item Execute the simulator of the AoK (with parameter $O(\varepsilon^{2})$) to produce a transcript, along with a valid witness $\tilde{m}_i$. Output $\bot$ if the AoK transcript is rejecting.

        By assumption, the transcript is $O(\varepsilon^{2})$-indistinguishable from a real execution, and the extracted witness is valid with probability $1- O(\varepsilon^{2})$.
        

        \item If $i\leq r$, simulate the chosen-input committed SFS protocol using $f_i(s, \tilde{m}_1, \dots \tilde{m}_i)$ as the function output. Output $\bot$ if the prover does not successfully terminate the protocol.

        By assumption, the trace distance between the simulated state and the real state is bounded by $O(\varepsilon^{2})$.
    \end{itemize}
    By a union bound, the above interaction is $\varepsilon$-indistinguishable from a real protocol run.

    Now, we define the final experiment which is the same as experiment $r$, except for the following difference.

    \begin{itemize}
        \item The verifier outputs $z = f_r(s,\Tilde{m}_1,\dots,\Tilde{m}_r)$ if it has not aborted.
    \end{itemize}

    To show that this is $\negl$-indistinguishable from experiment $r$, it suffices to bound the probability of extracting a set of messages $(\Tilde{m}_1, \dots, \tilde{m}_r)$ such that $f_r(s,\Tilde{m}_1,\dots,\Tilde{m}_r) \neq z$, but the final succinct AoK is accepting. Conditioned on successfully extracting all messages, which happens with noticeable probability, we claim that this can happen only with negligible probability. Assume towards contradiction that this is not the case, then running the extractor for the last AoK results in $(m_1^*, \dots, m_r^*)$ such that $f_r(s,m_1^*, \dots, m_r^*) = z$ with non-negligible probability. However, it must be the case that 
    \[
    (m_1^*, \dots, m_r^*) \neq (\Tilde{m}_1, \dots, \tilde{m}_r)
    \]
    since $f_r(s,\Tilde{m}_1,\dots,\Tilde{m}_r) \neq z$. This is a contradiction to the collision resistance of the commitment.
    
    To complete the proof, we note that we can turn the final experiment into an adversary $\mathcal{A}$ that plays against the original protocol, except that it can instruct the verifier to output $\bot$ with some probability. Indeed $\mathcal{A}$ can simply forward the extracted messages $(\Tilde{m}_1, \dots, \tilde{m}_r)$ to the verifier, and use its responses as input to the committed SFS protocol.





\end{proof}



\subsection{Putting Things Together: Proof of \cref{thm:compcom}}

We show how to combine the lemmas that we proved in order to establish \cref{thm:compcom}. First, \cref{lmm:beta-security} establishes the existence of a random-input committed SFS, assuming the existence of succinct commitments and state-preserving AoKs for NP. Without additional assumptions (since one-wayness is implied by collision-resistance) we can invoke \cref{lmm:random-chosen} to make the protocol chosen-input instead. Plugging this into the communication compression compiler (\cref{lmm:soundnes}) completes the proof of \cref{thm:compcom}.

Recall that the verifier in the resulting protocol is entirely classical, except for the ability to prepare and send uniformly random claw states. To make the verifier classical, we \ifllncs refer to the full version \cite{full}\else invoke \cref{thm:claw-state-RSP} from \cref{sec:rspv}\fi, which shows that OSP implies a simulation-secure notion of claw-state preparation\ifllncs\else~(\cref{def:RSPV})\fi. This simulation security enables a straightforward substitution of the quantum claw-state message with a classical verifier protocol. This implies the ``moreover'' part of the statement of \cref{thm:compcom}, where the verifier is entirely classical. 

%% file: full_protocol.tex
\section{Succinct Classically-Verifiable Arguments for QMA}
\label{sec:full_protocol_thm}

We state and prove the main theorem of this work.

\begin{theorem}
\label{thm:final}
   Assuming OSP and collapsing hash functions, there exists a succinct classically-verifiable $(1-\negl(\secp),\negl(\secp))$ argument system for QMA\ifllncs\else~(\cref{def:CVQC})\fi. 
\end{theorem}

\begin{proof}
    To prove the theorem, we consider the following final compilation of our building blocks. 

    Given a security parameter $\secp$, an instance $x \in \{0,1\}^*$, and an algorithm $C$ which is the verifier for a $\QMA$ promise problem $\cL = (\cL_\mathrm{yes}, \cL_\mathrm{no})$, do the following.
    \begin{enumerate}
        \item Apply our classical-verifier communication compression compiler (\cref{thm:compcom} in \cref{sec:compression_compiler}) to the round-efficient classically-verifiable argument for QMA (\ifllncs worked out in the full version \cite{full}\else Protocol~\ref{prot:main-blind} in \cref{sec:non-succinct_argument_for_qma_mnz24}\fi). Execute the compiled protocol on the inputs above. 

        \item Sequentially repeat Step 1 for $\lambda$ times. Output $\top$ if all repetitions output $\top$; output $\bot$ otherwise.
    \end{enumerate}

The completeness follows straightforwardly from the completeness of the round-efficient protocol \ifllncs\else(\cref{thm:round_efficient_classical_qma}) \fi and the completeness of the communication compression (\cref{thm:compcom}).

For soundness, note that \ifllncs the main protocol in the full version \else\cref{prot:main-blind} \fi only has constant soundness\ifllncs. \else, as stated in \cref{thm:round_efficient_classical_qma}.\fi  By the property of the communication compression compiler (\cref{thm:compcom}), the compiled protocol in Step 1 will inherit this constant soundness plus/minus an error of $1/\poly(\secp)$. Therefore, by sequentially repeating Step 1 and accepting only when all repetitions accept, we can obtain $\negl(\secp)$ soundness in a straightforward way.

Since \ifllncs the main protocol \else\cref{prot:main-blind} \fi has round complexity $\poly(\lambda)$, by the efficiency property of the compression compiler, the resulting protocol in Step 1 has communication complexity $\poly(\secp)$ and verifier runtime $|x| \cdot \poly(\secp)$. Therefore, the final protocol after $\secp$ sequential repetition preserves the communication complexity and verifier runtime.
\end{proof}
We also note that the same exact proof shows that applying our communication compression to the LWE-based constant-round classical verification of quantum sampling protocol (with $1/\poly$-soundness) from \cite{EC:CLLW} yields \emph{succinct} classical verification for quantum sampling, under the assumption of LWE (which additionally implies all assumptions used in \cref{thm:final}). That is, we establish the following.

\begin{theorem}\label{thm:sampling}
    Assuming LWE, there exists succinct classical verification of quantum sampling with inverse-polynomial soundness\ifllncs\else~(\cref{def:sampling})\fi.
\end{theorem}
{Moreover, \cref{thm:sampling} can be used to establish a stronger version of \cref{thm:compcom}, where the verifier is not assumed to be classical, as follows: First, observe that any interactive protocol between a prover and a verifier can be thought of as \emph{sampling} from a distribution specified by a quantum circuit where each ``gate'' is the unitary applied by the prover/verifier at a given round and each ``register'' consists of the internal state of a party or a message exchanged during the protocol. Now it is clear that, so long as the verifier's description corresponds to the honest algorithm, sampling from the output of this circuit is equivalent to running the original interactive protocol and returning the output of the verifier. Thus, assuming LWE, we can now apply \cref{thm:sampling} to sampling from such a distribution, obtaining a protocol with (i) a classical verifier and (ii) succinct communication.}

%% file: ack.tex
\section{Acknowledgments}

We thank Jiayu Zhang for valuable correspondence throughout the preparation of this work.

G.M.\ is supported by the European Research Council through an ERC Starting Grant (Grant agreement No.~101077455, ObfusQation) and funded by the Deutsche Forschungsgemleinschaft (DFG, German Research Foundation) under Germany's Excellence Strategy - EXC 2092 CASA – 390781972.

%% file: prelim.tex
\section{Preliminaries}
We denote by $[n]$ the set $\{1, \dots, n\}$. We say that a function~$f$ is negligible in the security parameter~$\lambda$ if~$f(\lambda) = {\lambda}^{-\omega(1)}$.

\subsection{Quantum Information}

In this section, we provide a brief overview of quantum information. For a more detailed introduction, we refer to \cite{nielsen2010quantum,watrous2018theory}.
A \emph{(quantum) register}~$A$ consisting of $n$ qubits is associated with the Hilbert space $\mathcal{H}_A = (\mathbb{C}^2)^{\otimes n}$. Given two registers $A$ and $B$, we denote the composite register by~$AB$.
The corresponding Hilbert space is given by the tensor product $\mathcal{H}_{AB} = \mathcal{H}_A \otimes \mathcal{H}_B$. The \emph{(quantum) state} of a register $A$ is described by a density operator $\rho_A$ on $\mathcal{H}_A$, which is a positive semi-definite Hermitian operator with trace equal to one.
A state is called \emph{pure} if it has rank one.
Thus, pure quantum states can be represented by unit vectors $\ket\psi_A \in \mathcal{H}_A$, with $\rho_A = \ketbra{\psi}{\psi}_A$.
For a quantum state~$\rho_{AB}$ on $\mathcal{H}_{AB}$, we denote $\rho_A = \Tr_B(\rho_{AB})\in \mathcal{H}_A$ the reduced state of $\rho_{AB}$ on $A$.

A \emph{projective measurement} is defined by a set of projectors $\{ \Pi_j \}_j$ such that $\sum_j \Pi_j = \mathsf{Id}$.
A projector $\Pi$ is a Hermitian operator such that $\Pi^2 = \Pi$, that is, an orthogonal projection.
Given a state $\rho$, the measurement yields outcome $j$ with probability $p_j = \Tr(\Pi_j \rho)$, upon which the state changes to ${\Pi_j \rho \Pi_j/p_j}$.
A basis measurement is one where~$\Pi_j=\ketbra{e_j}{{e_j}}$ and the $\{\ket{e_j}\}_j$ (necessarily) form an orthonormal basis.
A \emph{positive operator-valued measure} (POVM) is a generalization of a projective measurement.
A POVM is defined by a set of positive semi-definite operators~$\{\mathsf{E}_j\}_j$ such that~$\sum_j \mathsf{E}_j = \mathsf{Id}$ (that is, the $\mathsf{E}_j$ no longer need to be projections).
As before, given a quantum state~$\rho$, the probability of obtaining outcome $j$ when performing the measurement is given by~$p_j = \Tr(\mathsf{E}_j \rho)$, but the state after the measurement is no longer uniquely specified.

For two observables $A,B$, the commutator is denoted as $[A,B] := AB-BA$ and the anticommutator is $\{A,B\} := AB+BA$

The \emph{trace distance} between two states $\rho$ and $\sigma$, denoted by $\mathsf{TD}(\rho, \sigma)$, is defined as
\[
\mathsf{TD}(\rho, \sigma) = \frac{1}{2}\norm{\rho - \sigma}_1 = \frac{1}{2}\Tr\left(\sqrt{(\rho-\sigma)^\dagger(\rho - \sigma)}\right).
\]
The trace distance can be equivalently defined as the maximal advantage of any (possibly unbounded) quantum algorithm in distinguishing the two states $\rho$ and $\sigma$. Sometimes we use the shorthand $\rho \approx_\varepsilon \sigma$ to denote $\mathsf{TD}(\rho, \sigma) \leq \varepsilon$.

We recall the following version of the gentle measurement lemma.

\begin{lemma}[Gentle Measurement]\label{lmm:gentle}
    Let $\rho$ be a density matrix and $\Pi$ be an orthogonal projection. If $\Tr(\Pi\rho)\geq 1-\varepsilon$, then 
    \[
    \rho \approx_\varepsilon \frac{\Pi \rho \Pi}{\Tr(\Pi\rho)}.
    \]
\end{lemma}

We will also consider the notion of \emph{computational indistinguishability}.
Recall that a \emph{nonuniform QPT algorithm} $\mathcal A = \{\mathcal A_\lambda,\tau_\secp\}_\lambda$ consists of a family of quantum channels that can be implemented by a polynomial-size quantum circuit $\cA_\secp$ initialized with a quantum advice state $\tau_\secp$ of polynomial number of qubits (we will often drop the explicit mention of $\tau_\secp$ when referring to non-uniform QPT adversaries).
We call $\mathcal A$ a \emph{nonuniform QPT distinguisher} if the channel outputs a single bit.

\begin{definition}[Computational Indistinguishability]\label{def:computational indist}
We say that two families of states~$\{\rho_\lambda\}_\lambda$ and $\{\sigma_\lambda\}_\lambda$ are \emph{computationally indistinguishable} if for every nonuniform QPT distinguisher~$\mathcal A = \{\mathcal A_\lambda\}_\lambda$ there exists a negligible function~$\negl$ such that the following holds for all~$\lambda$:
\begin{align*}
     \abs{ \Pr\left( \mathcal A_\lambda(\rho_\lambda) = 1 \right) - \Pr\left( \mathcal A_\lambda(\sigma_\lambda) = 1 \right) } \leq \negl(\lambda).
\end{align*}
We abbreviate this definition by writing $\rho_\lambda \approx_c \sigma_\lambda$. If the indistinguishability only holds up to some $\epsilon = \epsilon(\secp)$, that is, for every nonuniform QPT distinguisher~$\mathcal A = \{\mathcal A_\lambda\}_\lambda$,
\begin{align*}
     \abs{ \Pr\left( \mathcal A_\lambda(\rho_\lambda) = 1 \right) - \Pr\left( \mathcal A_\lambda(\sigma_\lambda) = 1 \right) } \leq \epsilon,
\end{align*}
then we write $\rho_\secp \approx_{c,\epsilon} \sigma_\secp$.
\end{definition}

\subsection{Classically-Verifiable Argument Systems}

The class QMA consists of all (promise) languages that are decidable by a polynomial-size quantum circuit, on input a non-deterministically generated witness state. We recall the definition of a classical verification protocol for QMA.
\begin{definition}\label{def:CVQC}
    A classically-verifiable argument for QMA is an interaction between a QPT prover, with input an instance $x$ and state $\ket{\psi}$, and a PPT verifier, with input an instance $x$, \[\{\top,\bot\} \gets \langle P(1^\secp,\ket{\psi},x),V(1^\secp,x)\rangle,\] where $\{\top,\bot\}$ is the verifier's output. We say that the protocol is an $(\alpha(\secp),\beta(\secp))$ argument for QMA if for any language $\cL = (\cL_{\mathsf{yes}} \cup \cL_{\mathsf{no}})$ in QMA, the following hold.
    
    \begin{itemize}
        \item \textbf{Completeness.} For any $x \in \cL_{\mathsf{yes}}$, there exists a state $\ket{\psi}$ such that \[\Pr[\top \gets \langle P(1^\secp,\ket{\psi},x),V(1^\secp,x)\rangle] \geq \alpha(\secp).\]
        \item \textbf{Soundness.} For any $x \in \cL_{\mathsf{no}}$, and any QPT adversary $\{\cA_\secp\}_{\secp \in \bbN}$, \[\Pr\left[\top \gets \langle \cA_\secp,V(1^\secp,x)\rangle\right] \leq \beta(\secp).\]
    \end{itemize}
    Moreover, we say that the protocol is \emph{succinct} if there exists a fixed polynomial $p(\secp)$ such that the total communication is bounded by $p(\secp)$ and $V$'s run-time is bounded by $|x| \cdot p(\secp)$.
\end{definition}

Next, we define classical verification of quantum \emph{sampling}.

\begin{definition}\label{def:sampling}
    Classical verification of quantum sampling with inverse-polynomial soundness is an interaction between a QPT prover and a PPT verifier, with input the (potentially succinct) description $|Q|$ of a polynomial-size quantum circuit $\{Q_\secp\}_{\secp \in \bbN}$,
    \[z \gets \langle P(1^\secp,|Q|,\epsilon),V(1^\secp,|Q|,\epsilon)\rangle,\] where $z$ is the verifier's output. It should satisfy the following properties.
    \begin{itemize}
        \item \textbf{Completeness.} For any $\epsilon = 1/\poly(\secp)$, \[\mathsf{TV}\left(\left\{z \gets Q_\secp\right\},\left\{z \gets \langle P(1^\secp,|Q|,\epsilon),V(1^\secp,|Q|,\epsilon)\rangle\right\}\right) = \negl(\secp).\]
        \item \textbf{Soundness.} For any probability $p$, let $Z[Q,p]$ be the distribution that outputs 
        \begin{itemize}
            \item $\bot$ with probability $p$, and
            \item $z \gets Q_\secp$ otherwise.
        \end{itemize}
        For any $\epsilon = 1/\poly(\secp)$ and any QPT adversary $\{\cA_\secp\}_{\secp \in \bbN}$, there exists a $p$ such that for any QPT distinguisher $\{\D_\secp\}_{\secp \in \bbN}$, \[\bigg| \Pr[1 \gets \D_\secp\left(Z[Q,p]\right)] - \Pr[1 \gets \D_\secp\left(\langle \cA_\secp,V(1^\secp,|Q|,\epsilon)\rangle\right)]\bigg| \leq O(\epsilon).\]
    \end{itemize}
    Moreover, we say that the protocol is succinct if there exists a fixed polynomial $p(\secp)$ such that, for any $\epsilon = 1/\poly(\secp)$, the total communication is bounded by $p(\secp,1/\epsilon)$ and $V$'s run-time is bounded by $|Q| \cdot p(\secp,1/\epsilon)$.
\end{definition}

\subsection{Collision-Resistant and Collapsing Functions}
\label{sec:hash-functions}

Let $\cH = \{H_{\secp}\}_{\secp \in \N}$ be such that each $H_{\secp}$ is a distribution over functions $h \colon \{0,1\}^{\ell(\secp)} \to \{0,1\}^{\secp}$. We define the standard notion of collision resistance for functions.

\begin{definition}[Collision Resistance]
\label{def:crh}
$\cH$ is \emph{post-quantum collision resistant} if for every QPT adversary $\cA$,
\begin{equation*}
\Pr
\left[
\begin{array}{c}
x \neq x' \; \wedge \\
h(x) = h(x')  
\end{array}
\middle\vert
\begin{array}{r}
h \gets H_{\secp} \\
(x,x') \gets \cA(h)
\end{array}
\right]
= \negl(\secp).
\end{equation*}
\end{definition}
The notion of collapsing \cite{unruh2016computationally}, which we define next, generalizes collision resistance in the context of quantum algorithms.
\begin{definition}[Collapsing]
\label{def:collapsing}
$\cH$ is \emph{collapsing} if for every QPT adversary $\cA$,
\begin{equation*}
\Big|
\Pr[\mathsf{HCollasingExp}(0,\secp,\cA) = 1]
-
\Pr[\mathsf{HCollasingExp}(1,\secp,\cA) = 1]
\Big|
\leq \negl(\secp).
\end{equation*}
For $b \in \{0,1\}$ the experiment $\mathsf{HCollasingExp}(b,\secp,\cA)$ is defined as follows:
\begin{enumerate}
  \item The challenger samples $h \gets H_{\secp}$ and sends $h$ to $\cA$.
  \item $\cA$ replies with a (classical) binary string $y \in \{0,1\}^{\ell(\secp)}$ and a $\secp$-qubit quantum state on register $\cX$. (The requirement that $y$ is classical can be enforced by having the challenger immediately measure these registers upon receiving them.)
  \item The challenger computes $h$ in superposition on the $\secp$-qubit quantum state, and measures the bit indicating whether the output of $h$ equals $y$. If $h$ does not equal $y$, the challenger aborts and outputs $\bot$.
  \item If $b = 0$, the challenger does nothing. If $b = 1$, the challenger measures the $\secp$-qubit state in the standard basis.
  \item The challenger returns contents of the registers $\cX$ to $\cA$.
  \item $\cA$ outputs a bit $b'$, which is the output of the experiment.
\end{enumerate}
\end{definition}

\begin{claim}[\cite{unruh2016computationally}]
\label{claim:collapsing-crhf}
If $\cH$ is collapsing then $\cH$ is post-quantum collision resistant.
\end{claim}

\begin{remark} 
In an abuse of notation, we sometimes refer to collision resistant and/or collapsing functions as \textbf{succinct commitments}, where we use $\Com$ to refer to the description of the hash function $h$ (which is still sampled from some family $H_\secp$).
\end{remark}

\subsection{State-Preserving Arguments of Knowledge for NP}\label{sec:AoK}

We recall the definition of state-preserving arguments of knowledge (AoKs) \cite{BKL+,LMS}. It is shown in \cite{CMSZ,LMS} that collapsing hash functions are sufficient to construct state-preserving AoKs for NP.

\begin{definition}[State-Preserving AoKs for NP]\label{def:AoK}
    A state-preserving argument of knowledge (AoK) for NP is an interaction between a QPT prover and PPT verifier on input an instance $x$ \[\{\top,\bot\} \gets \langle P(1^\secp,x),V(1^\secp,x)\rangle,\] where $\{\top,\bot\}$ is the verifier's output. For any language $\cL$ in NP, the following properties hold.
    \begin{itemize}
    \item \textbf{Succinctness.} When invoked on a security parameter $\secp$ and instance size $n$ and a relation decidable in time $T$, the communication complexity of the protocol is $\poly(\secp, \log T )$. The verifier computational complexity is $\poly(\secp, \log T ) + \Tilde{O}(n)$.

    \item \textbf{Extraction.} There exists an extractor $\mathsf{Ext}(x, \varepsilon)$ with running in time $\poly(n, \secp, 1/\varepsilon)$ outputting a classical transcript $\tau$ and a classical string $w$, such that the following holds. For any $\epsilon = 1/\poly(\secp)$, any QPT adversarial prover $\cA = \{\cA_\secp\}_{\secp \in \bbN}$, any QPT distinguisher $\cD = \{\cD_\secp\}_{\secp \in \bbN}$, and any advice state $\rho_{A,E} = \{\rho_{\secp,A,E}\}_{\secp \in \bbN}$, consider the following games.\\
        
        \noindent\underline{$\ExpREAL[\cA,\rho]$}
        \begin{itemize}
            \item Run $\langle \cA(A, x),V(1^\secp,x)\rangle$ and let $\tau$ be the classical transcript and $A$ the residual state on the input register.
            \item Output $(\tau,A,E)$.
        \end{itemize}

        \noindent\underline{$\ExpIDEAL[\mathsf{Ext},\epsilon,\rho]$}
        \begin{itemize}
            \item Run $(\tau, w, A) \gets \mathsf{Ext}(x, \varepsilon)^{\mathcal{A}(A,x)}$.
            \item Output $(\tau,A,E)$.
        \end{itemize}
        Then,
        \[\bigg| \Pr[\cD(\ExpREAL[\cA,\rho]) = 1] - \Pr[\cD(\ExpIDEAL[\mathsf{Ext},\epsilon,\rho]) = 1]\bigg| \leq \epsilon.\]
        Furthermore, the probability that $\tau$ is an accepting transcript but $w$ is not a witness for $x$ is bounded by $\varepsilon + \negl(\secp)$. 
    \end{itemize}
\end{definition}

\subsection{Remote State Preparation}

We define remote state preparation with two flavors of security: (1) an indistinguishability-based notion, called oblivious state preparation (OSP), and (2) a simulation-based notion, called verifiable state preparation (VSP).

\begin{definition}[Oblivious State Preparation]\label{def:OSP}
   Oblivious state preparation (OSP) is a protocol that takes place between a PPT sender $S$ with input $b \in \{0,1\}$ and a QPT receiver $R$:
    \[(s,\ket*{\psi}) \gets \OSP\langle S(1^\secp,b),R(1^\secp)\rangle,\] where $s \in \{0,1\}$ is the sender's output and $\ket*{\psi}$ is the receiver's output. It should satisfy the following properties.
    \begin{itemize}
        \item \textbf{Correctness.} For any $b \in \{0,1\}$, let \[\Pi_b \coloneqq \sum_{s \in \{0,1\}}\ketbra*{s} \otimes H^b\ketbra*{s}H^b.\] Then for any $b \in \{0,1\}$, \[\E\left[\|\Pi_b\ket*{s}\ket*{\psi}\| : (s,\ket*{\psi}) \gets \OSP\langle S(1^\secp,b),R(1^\secp)\rangle\right] = 1-\negl(\secp).\]
        \item \textbf{Security.} For any QPT adversary $\{\cA_\secp\}_{\secp \in \bbN}$,
        \begin{align*}\Big|&\Pr\left[b_\cA = 0 : (s,b_\cA) \gets \OSP\langle S(1^\secp,0),\cA_\secp\rangle\right]\\ &- \Pr\left[b_\cA = 0 : (s,b_\cA) \gets \OSP\langle S(1^\secp,1),\cA_\secp\rangle\right]\Big|  = \negl(\secp).\end{align*}
    \end{itemize}
\end{definition}

\begin{definition}[Verifiable State Preparation]\label{def:RSPV}
    Verifiable state preparation (VSP) for a state family $\{\ket{\psi_i}\}_{i \in [K]}$ is a protocol that takes place between a PPT sender $S$ and a QPT receiver $R$:
    \[(i,\ket*{\psi_i}) \gets \VSP\langle S(1^\secp,\epsilon),R(1^\secp,\epsilon)\rangle,\] where $\epsilon$ is an error parameter, $i \in [K] \cup \{\bot\}$ is the sender's output, and $\ket*{\psi_i}$ is the receiver's output. It should satisfy the following properties.
    \begin{itemize}
        \item \textbf{Correctness.} For any $\epsilon = 1/\poly(\secp)$, \[\TD\left(\VSP\langle S(1^\secp,\epsilon),R(1^\secp,\epsilon)\rangle,\frac{1}{K}\sum_{i \in [K]}\ketbra{i} \otimes \ketbra{\psi_i}\right) = \negl(\secp).\]
        \item \textbf{Security.} For any $\epsilon = 1/\poly(\secp)$, any QPT adversarial receiver $\cA = \{\cA_\secp\}_{\secp \in \bbN}$, and any advice state $\rho_{A,E} = \{\rho_{\secp,A,E}\}_{\secp \in \bbN}$, there exists a QPT simulator $\Sim = \{\Sim_\secp\}_{\secp \in \bbN}$ such that for any QPT distinguisher $\cD = \{\cD_\secp\}_{\secp \in \bbN}$, the following holds. Define the following games.\\
        
        \noindent\underline{$\ExpREAL[\cA,\epsilon,\rho]$}
        \begin{itemize}
            \item Run $(i,A) \gets \VSP\langle S(1^\secp,\epsilon),\cA(A)\rangle$.
            \item Output $(i,A,E)$.
        \end{itemize}

        \noindent\underline{$\ExpIDEAL[\Sim,\epsilon,\rho]$}
        \begin{itemize}
            \item Run $(b,A) \gets \Sim(1^\secp,\epsilon,A)$.
            \item If $b = 1$, sample $i \gets [K]$, otherwise if $b = 0$, set $i = \bot$ and $\ket{\psi_i} = 0$.
            \item Run $A \gets \Sim(1^\secp,\epsilon,A,\ket{\psi_i})$.
            \item Output $(i,A,E)$.
        \end{itemize}

        Then,
        \[\bigg| \Pr[\cD(\ExpREAL[\cA,\epsilon,\rho]) = 1] - \Pr[\cD(\ExpIDEAL[\Sim,\epsilon,\rho]) = 1]\bigg| \leq \epsilon.\]
    \end{itemize}
\end{definition}

%% file: remote_state_preparation.tex
\section{Verifiable Preparation of Claw States}
\label{sec:rspv}

In this section, we formalize the fact that OSP (\cref{def:OSP}) implies VSP (\cref{def:RSPV}) for claw states. That is, we prove the following theorem.

\begin{theorem}[VSP for Claw States]\label{thm:claw-state-RSP}
    Assuming OSP, there exists VSP for the ``claw state'' family \[\left\{\frac{1}{\sqrt{2}}\left(\ket{0,x_0}+\ket{1,x_1}\right)\right\}_{x_0,x_1 \in \{0,1\}^{m(\secp)}}\] for any polynomial $m(\secp)$.
\end{theorem}

This follows by combining the following two lemmas.

\begin{lemma}\label{lemma:OSP-VSP}
    Assuming OSP, there exists VSP for the ``BB84 state'' family \[\left\{H^\theta\ket{x}\right\}_{x,\theta \in \{0,1\}^{n(\secp)}}\] for any polynomial $n(\secp)$. 
\end{lemma}

\begin{proof}
    First, due to sequential repetition for VSP (\cite[Theorem 3.3]{ZhaRSP}), it suffices to show the existence of VSP for the state family $\{H^\theta\ket{x}\}_{x,\theta \in \{0,1\}}$. For this, consider the following pair of protocols.\\

    \noindent\underline{$\VSP_\Test$}
    \begin{itemize}
        \item The sender $S$ samples $\theta \gets \{0,1\}$.
        \item Run $(x,\ket{\psi}) \gets \OSP\langle S(1^\secp,\theta),R(1^\secp)\rangle$.
        \item The sender $S$ samples $q \gets \{0,1\}$, and sends $q$ to $R$.
        \item $R$ measures their state $\ket{\psi}$ in the $X+Z$ basis if $q = 0$ or the $X-Z$ basis if $q = 1$ to obtain a bit $a$. $R$ sends $a$ to $S$.
        \item $S$ outputs 1 if $a = x \oplus \theta \cdot q$.
    \end{itemize}

    \noindent\underline{$\VSP_\Out$}
    \begin{itemize}
        \item The sender $S$ samples $\theta \gets \{0,1\}$.
        \item Run $(x,\ket{\psi}) \gets \OSP\langle S(1^\secp,\theta),R(1^\secp)\rangle$.
        \item $R$ outputs $\ket{\psi}$ and $S$ outputs $(x,\theta)$.
    \end{itemize}
    By the arguments in \cite[Section 3.5.2]{ZhaRSP} (amplification of ``PrePSPV-with-score'' to RSPV), in order to establish $\VSP$ for $\{H^\theta\ket{x}\}_{x,\theta \in \{0,1\}}$, it suffices to show that for any $\epsilon = 1/\poly(\secp)$, there exists a $\delta = 1/\poly(\secp)$ such that the following holds. For any QPT adversarial receiver $\cA = \{\cA_\secp\}_{\secp \in \bbN}$ and advice state $\rho_{A,\sE} = \{\rho_{\secp,A,E}\}_{\secp \in \bbN}$ such that 
    \[\Pr\left[1 \gets \VSP_\Test\langle S(1^\secp),\cA(A)\rangle\right] \geq \cos^2(\pi/8)-\delta,\]
    there exists a QPT simulator $\Sim = \{\Sim_\secp\}_{\secp \in \bbN}$ such that for any QPT distinguisher $\cD = \{\cD_\secp\}_{\secp \in \bbN}$, 
    \begin{align*}
        \bigg|&\Pr\left[\cD\left((x,\theta),A,E\right) = 1 : ((x,\theta),A) \gets \VSP_\Out\langle S(1^\secp),\cA(A)\rangle\right] \\
        &- \Pr\left[\cD\left((x,\theta),A,E\right) = 1 : \begin{array}{r}x,\theta \gets \{0,1\} \\ A \gets \Sim(1^\secp,A,H^\theta\ket{x})\end{array}\right]\bigg| \leq \epsilon.
    \end{align*}

    To see this, we will adapt arguments made in \cite[Section 4.2]{ZhaRSP}. We will make use of the following key lemma. First, we introduce the following notation.
    \begin{itemize}
        \item For $k \in [0,7]$, let $\ket{\phi_k} \coloneqq \cos(k\pi/8)\ket{0} + \sin(k\pi/8)\ket{1}$.
        \item Define the following ``CHSH acceptance'' projectors $\Pi^{(\theta,q)}_{B,X}$, where $\theta \in \{0,1\}$ denotes Alice's question and $q \in \{0,1\}$ denotes Bob's question. Each projector  operates on two single-qubit registers $B,X$, where $B$ is meant to hold Bob's qubit and $X$ is meant to hold Alice's answer $x$.
        \begin{align*}
        &\Pi^{(0,0)}_{B,X} = \ketbra{\phi_1} \otimes \ketbra{0} + \ketbra{\phi_5} \otimes \ketbra{1}  \\
        &\Pi^{(0,1)}_{B,X} = \ketbra{\phi_7} \otimes \ketbra{0} + \ketbra{\phi_3} \otimes \ketbra{1}  \\
        &\Pi^{(1,0)}_{B,X} = \ketbra{\phi_1} \otimes \ketbra{0} + \ketbra{\phi_5} \otimes \ketbra{1} \\
        &\Pi^{(1,1)}_{B,X} = \ketbra{\phi_3} \otimes \ketbra{0} + \ketbra{\phi_7} \otimes \ketbra{1}  
        \end{align*}
    \end{itemize}

    \begin{lemma}\label{lemma:CHSH}
        Let $X$ be a classical register, $Z$ be an arbitrary-size quantum register, and $B$ be a single-qubit register. Let $\rho^{(0)}_{Z,B,X}, \rho^{(1)}_{Z,B,X}$ be any two states such that \[\Tr_X\left(\rho^{(0)}\right)  \approx_c \Tr_X\left(\rho^{(1)}\right)\] and 
        \[\frac{1}{4}\Tr\left(\Pi^{(0,0)}_{B,X}\rho^{(0)}\right) + \frac{1}{4}\Tr\left(\Pi^{(0,1)}_{B,X}\rho^{(0)}\right) + \frac{1}{4}\Tr\left(\Pi^{(1,0)}_{B,X}\rho^{(1)}\right) + \frac{1}{4}\Tr\left(\Pi^{(1,1)}_{B,X}\rho^{(1)}\right) \geq \cos^2(\pi/8) - \delta.\] 

        Then, there exists $\sigma_Z$ such that 
        \[\left(\frac{1}{2}\rho^{(0)} \otimes \ketbra{0}_T + \frac{1}{2}\rho^{(1)} \otimes \ketbra{1}_T\right) \approx_{c,\poly(\delta)} \sigma_Z \otimes \frac{1}{4}\sum_{x,\theta \in \{0,1\}} \left(H^\theta\ketbra{x} H^\theta\right)_B \otimes \ketbra{x}_X \otimes \ketbra{\theta}_T  .\]
    \end{lemma}

    \begin{proof}
        Write \[\rho^{(0)}_{Z,B,X} = \rho^{(0,0)}_{Z,B} \otimes \ketbra{0}_X + \rho^{(0,1)}_{Z,B} \otimes \ketbra{1}_X\] and \[\rho^{(1)}_{Z,B,X} = \rho^{(1,0)}_{Z,B} \otimes \ketbra{0}_X + \rho^{(1,1)}_{Z,B} \otimes \ketbra{1}_X.\]

        By a standard argument, we have that for any $\tau_B$,
        \[\frac{1}{2}\Tr\left(\ketbra{\phi_1}\tau\right) + \frac{1}{2}\Tr\left(\ketbra{\phi_7}\tau\right) \leq \cos^2(\pi/8),\]
        this expression is maximized at $\tau = \ketbra{0}$, and thus there exists $\sigma_Z^{(0,0)}$ such that
        \[\rho_{Z,B}^{(0,0)} \approx_{\poly(\delta)} \sigma^{(0,0)}_Z \otimes \ketbra{0}_B.\] Analogous arguments allow us to conclude that there exist $\sigma_Z^{(0,1)},\sigma_Z^{(1,0)},\sigma_Z^{(1,1)}$ such that 
        \[\rho^{(0)}_{Z,B,X} \approx_{\poly(\delta)} \sigma^{(0,0)}_Z \otimes \ketbra{0}_B \otimes \ketbra{0}_X + \sigma^{(0,1)}_Z \otimes \ketbra{1}_B \otimes \ketbra{1}_X\] and \[\rho^{(1)}_{Z,B,X} \approx_{\poly(\delta)} \sigma^{(1,0)}_Z \otimes \ketbra{+}_B \otimes \ketbra{0}_X + \sigma^{(1,1)}_Z \otimes \ketbra{-}_B \otimes \ketbra{1}_X.\]

        To complete the proof, it remains to show that 
        \begin{enumerate}   
            \item $\sigma_Z^{(0,0)} \approx_{c,\poly(\delta)} \sigma_Z^{(0,1)}$,
            \item $\sigma_Z^{(1,0)} \approx_{c,\poly(\delta)} \sigma_Z^{(1,1)}$,
            \item and $\sigma_Z^{(0,0)} \approx_{c,\poly(\delta)} \sigma_Z^{(1,0)}$,
        \end{enumerate}
        so that we can set $\sigma_Z \coloneqq \sigma_Z^{(0,0)}$. Claim (3) follows directly by combining Claim (1), Claim (2), and the first condition of the lemma statement. Thus it remains to prove Claim (1) and Claim (2).
        
        For Claim (1), consider any QPT $\cD$ and define \[\epsilon \coloneqq \Pr\left[\cD\left(\sigma_Z^{(0,0)}\right) = 1\right] - \Pr\left[\cD\left(\sigma_Z^{(0,1)}\right) = 1\right].\]
        We define a $\cD'$ that operates on registers $Z$ and $B$ as follows.
        \begin{itemize}
            \item Run $\cD$ on $Z$ and apply a bit flip to $B$ if the output is 1.
		\item Measure $B$ in the standard basis and output 1 if the result is 0.
        \end{itemize}
        We have that \[\Pr[\cD'\left(\Tr_X\left(\rho^{(0)}\right)\right) = 1] = 1/2 \pm \poly(\delta) + \poly(\epsilon),\] while \[\Pr[\cD'\left(\Tr_X\left(\rho^{(1)}\right)\right) = 1] = 1/2 \pm \poly(\delta).\] Thus, by the first condition in the lemma statement, we conclude that $\epsilon = \poly(\delta)$, completing the proof of Claim (1).

        For Claim (2), consider any QPT $\cD$ and define \[\epsilon \coloneqq \Pr\left[\cD\left(\sigma_Z^{(1,0)}\right) = 1\right] - \Pr\left[\cD\left(\sigma_Z^{(1,1)}\right) = 1\right].\]
        We define a $\cD'$ that operates on registers $Z$ and $B$ as follows.
        \begin{itemize}
            \item Run $\cD$ on $Z$ and apply a phase flip to $B$ if the output is 1.
		\item Measure $B$ in the Hadamard basis and output 1 if the result is $\ket{+}$.
        \end{itemize}
        We have that \[\Pr\left[\cD'\left(\Tr_X\left(\rho^{(1)}\right)\right) = 1\right] = 1/2 \pm \poly(\delta) + \poly(\epsilon),\] while \[\Pr[\cD'\left(\Tr_X\left(\rho^{(0)}\right)\right) = 1] = 1/2 \pm \poly(\delta).\] Thus, by the first condition in the lemma statement, we conclude that $\epsilon = \poly(\delta)$, completing the proof of Claim (2), and hence the lemma.



        

    \end{proof}

    Now, we complete the proof of \cref{lemma:OSP-VSP}. Fix any adversarial receiver $\cA$. Let $\tau_A$ be their state after the $\OSP$ is completed, and before the question $q$ is sent. Let $O_0,O_1$ be the binary observables they apply to register $A$ (where we have expanded $A$ sufficiently so that $O_0$ and $O_1$ are projective) in order to respond to questions $q = 0$ and $q=1$ respectively. By the proof of \cite[Theorem 6.2]{bartusek2025power} combined with \cite[Theorem 4.7]{BGKPV}, we have that \[\Tr\left(\{O_1,O_2\}^2 \tau\right) = O(\delta),\] where $\{O_0,O_1\} = O_0O_1 + O_1O_0$. Then, by a standard argument (see e.g. \cite[Claim 6.4]{Vid20-course}), there exists an efficient isometry $V: A \to B,A'$ that makes calls to $O_0$ and $O_1$ such that \begin{align*}O_0(\tau) &\approx_{\poly(\delta)} V^\dagger\left(\ketbra{\phi_1}_B - \ketbra{\phi_5}_B\right)V(\tau) \\ O_1(\tau) &\approx_{\poly(\delta)} V^\dagger\left(\ketbra{\phi_3}_B - \ketbra{\phi_7}_B\right)V(\tau).\end{align*} Now, we can define the simulator.\\
    
    \noindent\underline{$\Sim(1^\secp,A,H^\theta\ket{x})$}
    \begin{enumerate}
    	\item Sample $\theta \gets \{0,1\}$ and run $((x,\theta),A) \gets \OSP\langle S(1^\secp,\theta),\cA(A)\rangle$.
	\item Apply $V$, swap the state on register $B$ with $H^\theta\ket{x}$, and apply $V^\dagger$. 
	\item Output the resulting state on register $A$.
    \end{enumerate}

    Define register $Z \coloneqq (E,A')$, let $X$ be the single-qubit (classical) register holding $S$'s output $x$, and let $T$ be the single-qubit (classical) register holding $S$'s output $\theta$. Let $\rho^{(0)}_{Z,B,X} = (E,V(A),X)$ conditioned on $\theta = 0$ and $\rho^{(1)}_{Z,B,X} = (E,V(A),X)$ conditioned on $\theta = 1$. Then by the security of $\OSP$ and the fact that $\cA$ succeeds with probability $\cos^2(\pi/8) - \delta$ in $\VSP_\Test$, we have that $\rho^{(0)}$ and $\rho^{(1)}$ satisfy the conditions of \cref{lemma:CHSH}. By the conclusion of \cref{lemma:CHSH}, we have that 
    \[\left((x,\theta),A,E\right) : ((x,\theta),A) \gets \VSP_\Out\langle S(1^\secp),\cA(A)\rangle\] is $\poly(\delta)$ indistinguishable from a state where the application of Step 2 of $\Sim$ has no effect, which completes the proof.

\end{proof}

The following is established in \cite[Theorem 6.7]{ZhaRSP}

\begin{lemma}[\cite{ZhaRSP}]
    For any polynomial $m(\secp)$, there exists a polynomial $n(\secp)$ such that the following holds. Assuming VSP for the BB84 state family \[\left\{H^\theta\ket{x}\right\}_{x,\theta \in \{0,1\}^{n(\secp)}},\] there exists VSP for the claw state family \[\left\{\frac{1}{\sqrt{2}}\left(\ket{0,x_0}+\ket{1,x_1}\right)\right\}_{x_0,x_1 \in \{0,1\}^{m(\secp)}}.\]
\end{lemma}

%% file: non_succinct_qma.tex
\section{Round-Efficient Classically-Verifiable Arguments for QMA}\label{def:compiled-game}
\label{sec:non-succinct_argument_for_qma_mnz24}
In this section, we derive a classically-verifiable QMA verification protocol with a fixed $\poly(\secp)$ number of rounds, assuming OSP. Our round-efficient protocol is an adapted and simplified version of the question-succinct QMA verification protocol in \cite{metger2024succinct} (Protocol 2 in their paper). We will adapt their protocol so that it remains round efficient, but remove the need for the question-succinctness property. At a high level, we will modify the protocol in the following two ways: 
\begin{enumerate}
    \item Replace the QFHE scheme with a classical-verifier blind delegation scheme built from OSP \cite{bartusek2025power}, so that we don't need to assume LWE.
    \item Remove the components needed for question-succinctness, since we will make the protocol succinct later using our communication compression compiler from \cref{sec:compression_compiler}.
\end{enumerate}

\subsection{Preliminaries}
We first present some preliminaries used in the protocol and proofs for this section.
\subsubsection{Pauli matrices and the Heisenberg-Weyl group}
We use the usual Pauli matrices $\sigma_X, \sigma_Y, \sigma_Z$. We will also find it convenient to set $\sigma_{\1} = \1$.
For $w \in \{\1, X,Z\}^n$ and $a \in \bits^n$, we define
\begin{align*}
\sigma_w(a) = \bigotimes_i \sigma_{w_i}^{a_i} \,.
\end{align*}
We also write $\sigma_w = \sigma_w(\vec 1) = \otimes_i \sigma_{w_i}$.

In addition, we define the Pauli \emph{projections} $\pi^w_a$ as
\begin{align}
\label{eq:pauli-proj}
\pi^w_u = \bigotimes_i \left( \frac{\1 + (-1)^{u_i} \sigma_{w_i}}{2} \right) = \E_{a \in \bits^n} (-1)^{u \cdot a} \sigma_w(a)\,.
\end{align}

\paragraph{Measurement notations.} We will need to establish some notation specific to the prover's ``Bob'' operators in Protocol~\ref{prot:main-2-prover}.

\begin{definition}[Projective measurements]
~
\begin{enumerate}
\item \textbf{Pure basis measurements.} In the Pauli braiding test (Protocol~\ref{prot:pauli-braiding}), and also in the $b=0$ case of the mixed versus pure basis test (Protocol~\ref{prot:mixed-vs-pure}), Bob is asked a single-bit question corresponding to a basis ($X$ or $Z$). For each basis $W \in \{X,Z\}$, we notate Bob's projective measurement after receiving question $W$ as a set of projectors $\{P^{W}_u\}_{u}$ with outcomes $u \in \{0,1\}^n$.
\item \textbf{Mixed basis measurements.} In the mixed versus pure basis test (Protocol~\ref{prot:mixed-vs-pure}), Bob receives a question $w \in \{\1, X, Z\}^n$ corresponding to $n$ Pauli bases (one for each of $n$ qubits). For each Pauli string $w \in \{\1, X,Z\}^n$, we notate Bob's projective measurement after receiving question $w$ as a set of projectors $\{M^{w}_u\}_{u \in \bits^n}$.
We can assume without loss of generality that Bob always answers ``0'' on indices where he was asked to measure the identity, i.e.~formally we can assume that $M^w_u = 0$ if there exists an index $i \in [n]$ for which $w_i = \1$ but $u_i = 1$.
The reason that this assumption is without loss of generality is that we can always replace Bob's measurements by a post-processed version that has this property; since the verifier ignores all indices for which $w_i = \1$, this post-processing affects neither Bob's success probability nor any of the rigidity statements we show below.
\end{enumerate}
\end{definition}

\begin{definition}[Pure basis observables]
    For $W \in \{X, Z\}$ and $a \in \bits^n$, define the binary observable
    \[ W(a) := \sum_{u} (-1)^{ u \cdot a} P^W_u. \]
\end{definition}
\begin{definition}[Mixed basis observables]
For $w \in \{\1, X,Z\}^n$ and $a \in \bits^n$, define the binary observables 
    \begin{align*}
        O^w(a) &= \sum_{u \in \bits^n} (-1)^{a \cdot u} M^w_u \,,\\
        O^w_W(a) &= \sum_{u \in \bits^n} (-1)^{\vec 1_{w = W \wedge a = 1} \cdot u} M^w_u \,.
    \end{align*}
We also write $O^w = O^w(\vec 1)$ and $O^w_W = O^w_W(\vec 1)$.
\end{definition}

\paragraph{Alice and Bob's measurements.} For convenience, we summarize the different question types that Alice and Bob may each see in Protocol~\ref{prot:main-2-prover} here. (These are not necessarily in one-to-one correspondence with the subgames of the protocol, since some subgames are indistinguishable from Bob's point of view.)

We denote the set of Alice questions by $\cQ_A$ and the set of Bob questions by $\cQ_B$.

\paragraph{Alice questions:}
\begin{itemize}
\item $(a, b)$: $a$ and $b$ are $n$-bit strings and the operations are measuring $\sigma_Z(a)$ and $\sigma_X(b)$.
\item $(\mathsf{MS}, a, b, i_1, i_2, i_3)$:  play magic square for the cells indicated by $i_1, i_2, i_3$ with the operators for cells 1 and 5 coinciding with the $\sigma_X(a)$ and $\sigma_Z(b)$ operators.
\item $X$ or $Z$: measure all qubits in $\sigma_X$ or $\sigma_Z$ basis.
\item $w$ for Pauli string $w \in \{\1,X,Z\}^n$: do mixed basis measurement in Pauli bases given by $w$.
\item $\mathsf{tele}$: do teleportation measurement.
\end{itemize}

\paragraph{Bob questions:}
\begin{itemize}
\item $X$ or $Z$: measure all qubits in $\sigma_X$ or $\sigma_Z$ basis.
\item $(\mathsf{MS}, a, b, j)$: play magic square for the cell indicated by $j$ with the operators for cells 1 and 5 coinciding with the $\sigma_X(a)$ and $\sigma_Z(b)$ operators.
\item $w$ for Pauli string $w \in \{\1,X,Z\}^n$: do mixed basis measurement in Pauli bases given by $w$.
\end{itemize}

\subsubsection{State dependent norm}
Let $\mathsf{Pos}$ denote the set of positive semidefinite
$A: \cH \to \cH$.

\begin{definition}[State-dependent inner product and norm]\label{def:state_dep_inner_product}
 Let $\cH$ be a finite-dimensional Hilbert space and $A, B \in L(\cH)$ be linear operators on $\cH$. Let $\psi \in \mathsf{Pos}(\cH)$. We define the state-dependent (semi) inner product of $A$ and $B$ w.r.t $\psi$ as 
\begin{equation*}
\langle A, B \rangle_\psi = \Tr[ A^\dagger B \psi ] \,.
\end{equation*} 
This induces the state-dependent (semi) norm 
\begin{equation*}
\norm{A}_\psi^2 = \langle A, A \rangle_\psi = \Tr[ A^\dagger A \psi ] \,.
\end{equation*}
\end{definition}

\begin{remark} \label{rem:schatten_norm}
The state dependent (semi) norm can also be expressed as a Schatten 2-norm (also called the Hilbert-Schmidt norm): 
\begin{equation*}
\norm{A}_\psi = \norm{A \psi^{1/2}}_2 \,.
\end{equation*}
\end{remark}

The following are some standard and easy-to-prove properties about the norm, taken from \cite{metger2024succinct}. 
\begin{lemma}[Basic properties of the state dependent norm] \label{lem:state_dep_norm_props}
For all (not necessarily normalised) states $\psi, \psi' \in \mathsf{Pos}(\cH)$ and linear operators $A, B \in L(\cH)$ on some finite-dimensional Hilbert space $\cH$:
\begin{enumerate}
\item $\norm{A}_{B \psi B^\dagger} = \norm{AB}_{\psi}$. \label{item:conj_mult}
\item $\norm{A B}_{\psi} \leq \norm{A}_{\infty} \norm{B}_{\psi}$, where $\Vert A \Vert_\infty : = \sup_{\Vert v \Vert = 1} \Vert Av \Vert$
\label{item:submult_infty}
\item For any unitary $U$, $\norm{UA}_\psi = \norm{A}_\psi$. \label{item:left_unitary_inv}
\item Linearity of the squared norm in the state: $\norm{A}_{\psi + \psi'}^2 = \norm{A}_{\psi}^2 + \norm{A}_{\psi'}^2$. \label{item:linearity}
\item Triangle inequality for the squared norm: $\norm{A+B}_\psi^2 \leq 2 \norm{A}_\psi^2 + 2 \norm{B}_\psi^2$. \label{item:triangle_ineq_sq}
\end{enumerate}
\end{lemma}

\begin{lemma}[Lemma 2.9 in \cite{metger2024succinct}]
\label{lem:state_replace}
For an observable $A$ on $\cH$ and two states $\psi, \psi'$ on $\cH$ with $\norm{\psi - \psi'}_1 \leq \eps$, we have that %
\begin{align*}
\norm{A}_{\psi}^2 \approx_{\norm{A}_{\infty}^2 \eps} \norm{A}_{\psi'}^2 \,.
\end{align*}
\end{lemma}

\input{hamiltonian_prelims}

\input{blind_delegation_nonlocal}

\subsection{A Non-Local Game for $\QMA$ with Constant Gap}
\label{sec:nonsuccinct_qma_protocol}

We present a family of non-local games for verifying QMA resulting from minor modifications made to the protocol from \cite{metger2024succinct}. The protocol in \cite{metger2024succinct} is question-succinct, but we no longer require this property, since our compiler from Section~\ref{sec:compression_compiler} will restore succinctness. Therefore, we make the \cite{metger2024succinct} protocol information-theoretic by removing the use of PRG to sample the randomness used for challenges. In particular, we make the following changes from \cite{metger2024succinct}:
\begin{enumerate}
    \item Use uniformly random $a$,$b$ in Protocols \ref{prot:comgame} and \ref{prot:mixed-vs-pure}, rather than sampling them from a particular distribution with small support.
    \item In the Hamiltonian test (Protocol~\ref{prot:hamiltonian}), the verifier  does not need to use succinct randomness for a subsampled Hamiltonian as specified in \cite[Protocol 8]{metger2024succinct}. Instead, it simply uses the regular distribution $D$ over $\{\1,X,Z\}^{n}$ for testing the X/Z Hamiltonian problem. %
\end{enumerate}
\begin{longfbox}[breakable=false, padding=1em, padding-right=1.8em, padding-top=1.2em, margin-top=1em, margin-bottom=1em]
\begin{protocol} 
\label{prot:main-2-prover}
Non-local game for $\QMA$
\end{protocol}
\textbf{Inputs:} An instance $x \in \{0,1\}^*$ and an algorithm $C$ which is the verifier for a $\QMA$ promise problem $\cL = (\cL_\mathrm{yes}, \cL_\mathrm{no})$ such that $\cL \in \mathsf{QMA}$.

On input $x, C$, execute the reduction given in \cref{thm:qma-to-hamiltonian} to produce a Hamiltonian problem $(H, \alpha(n), \beta(n))$ (see \cref{def:hamiltonian-problem}) with $\beta(n) - \alpha(n) = 1 - \mathsf{negl}(n)$, where $H$ is an $n$-qubit operator ($n = \poly(|x|)$) with the following form:
\begin{align*}
H =  \1 - \underset{{w \in D}}{\E}\sum_{u \in Q(w)} \pi^w_u
\end{align*}
where
\begin{itemize}
    \item $w \in \{\1, X, Z\}^n$ is a Pauli string,
    \item $\{\pi^w_u\}_{u}$ is the projective measurement corresponding to measuring $n$ qubits in the natural way in the Pauli bases specified by $w$ (see \cref{eq:pauli-proj} for a formal definition),
    \item $Q(w) \subseteq \{0,1\}^n$ is a set for which membership can be decided in polynomial time in $n$ given $w$, and
    \item 
    a distribution over $\{\1,X,Z\}^{n}$ which can be efficiently sampled.
\end{itemize}
In words, if $H$ is of this form, then the quantity $1 - \Tr[ H \rho ]$ for any $n$-qubit state $\rho$ can be estimated by a QPT (in $n$) verifier who samples a $w$ from $D$, measures the $n$ qubits of $\rho$ in the Pauli bases specified by $w$, obtains outcomes $a$, and accepts iff $a \in Q(w)$.

\emph{Assume that we are considering a yes-instance $x \in \cL_{\mathrm{yes}}$. Honest Alice receives a witness that $x \in \cL_{\mathrm{yes}}$ as input, and converts it in polynomial time to a $n$-qubit witness $\rho$ that the ground energy of $H$ is $\leq \alpha(n)$. Honest Alice and Bob then share $n+1$ EPR pairs between themselves; most of the tests in the following protocol use only the last $n$ EPR pairs, but the first EPR pair will be used in Protocol~\ref{prot:acomgame}.}

The verifier executes each of the following tests with Alice and Bob with equal probability.
\begin{enumerate}
    \item \textbf{Pauli braiding test.} Described in Protocol~\ref{prot:pauli-braiding}; $n$ will be the number of qubits on which $H$ acts.
    \item \textbf{Mixed-versus-pure basis test.} Described in Protocol~\ref{prot:mixed-vs-pure}; the distribution $D$ will be the $D$ such that $H = \1 -  \underset{{w \sim D}}{\E}\sum_{u \in Q(w)} \pi^w_u$.
    \item \textbf{Hamiltonian test.} Described in Protocol~\ref{prot:hamiltonian}; the distribution $D$ and the sets $\{Q(w)\}_w$ will be those such that $H = \1 - \underset{{w \sim D}}{\E}\sum_{u \in Q(w)} \pi^w_u$.
\end{enumerate}
\end{longfbox}

\begin{longfbox}[breakable=false, padding=1em, padding-right=1.8em, padding-top=1.2em, margin-top=1em, margin-bottom=1em]
\begin{protocol} Pauli braiding test
\label{prot:pauli-braiding} \end{protocol}
\textbf{Input:} An integer $n$.

The verifier picks $a,b$ uniformly at random from $\{0,1\}^n$. 
\begin{enumerate}
\item \textbf{(Commutation):} If $a \cdot b = 0$, the verifier executes the commutation test (Protocol~\ref{prot:comgame}) with questions $a, b$.
\item \textbf{(Anticommutation):} If $a \cdot b = 1$, the verifier executes the anticommutation test (Protocol~\ref{prot:acomgame}) with questions $a, b$.
\end{enumerate}
\end{longfbox}

\begin{longfbox}[breakable=false, padding=1em, padding-right=1.8em, padding-top=1.2em, margin-top=1em, margin-bottom=1em]
\begin{protocol} Commutation test
\label{prot:comgame} \end{protocol}
\textbf{Input:} Questions $a$ and $b$ in $\{0,1\}^n$.
\begin{enumerate}
\item The verifier sends $(a, b)$ %
to Alice, and receives responses $(u_a, u_b)$ with $u_a, u_b \in \{0,1\}$. \emph{Honest Alice 
measures $\sigma_Z(a)$ and $\sigma_X(b)$ on her last $n$ qubits, and returns the results as $u_a$ and $u_b$ respectively.}
\item The verifier picks $W \in \{X,Z\}$ uniformly at random and sends $W$ to Bob. Bob responds with $v \in \{0,1\}^n$. Note that Bob is \emph{not} told that he is playing the commutation test: he receives only the question label $W$. \emph{Honest Bob measures his last $n$ qubits in the $W$ basis.}
\item If $W=Z$, the verifier accepts iff $\prod_{i: a_i = 1} v_{i} = u_a$. If $W=X$, the verifier accepts iff $\prod_{i: b_i = 1} v_{i} = u_b$.
\end{enumerate}
\end{longfbox}

\begin{longfbox}[breakable=false, padding=1em, padding-right=1.8em, padding-top=1.2em, margin-top=1em, margin-bottom=1em]
\begin{protocol} Anticommutation test 
\label{prot:acomgame} \end{protocol}
\textbf{Input:} Questions 
$a$ and $b$ where $a,b \in \{0,1\}^n$.

In this protocol, the verifier plays a version of the Mermin-Peres Magic Square game~\cite{aravind2002simple,mermin1990simple,peres1990incompatible} with the provers, in which Alice is asked to measure three observables forming a row or column of the square, and Bob is asked to measure an observable from a single cell of the square. Both provers are instructed to use the observables labeled by $a$ and $b$ for two specific cells in the square (the top centre and centre left cells), as indicated below.

\begin{enumerate}
\item The verifier chooses a cell index $j \in [9]$ uniformly at random; it then chooses uniformly at random a row or a column on a $3 \times 3$ grid which contains cell $j$. Suppose that the cell indices of the 3 cells in this row or column are $(i_1, i_2, i_3)$ (one of these will be equal to $j$).
\item The verifier sends $(\mathsf{MS}, a, b, i_1, i_2, i_3)$ to Alice, and receives responses $u_1, u_2, u_3 \in \{0,1\}$. \emph{Honest Alice gets $a$ and $b$. Then she measures the three $(n+1)$-qubit observables associated with cells $i_1, i_2, i_3$ in the following grid, and returns all three 1-bit outcomes to the verifier:}

\begin{tabular}{c|c|c}
$\sigma_Z \otimes \1$ & $\1 \otimes \sigma_Z(a)$ & $\sigma_Z \otimes \sigma_Z(a)$ \\
\hline
$\1 \otimes \sigma_X(b)$ & $\sigma_X \otimes \1$ & $\sigma_X \otimes \sigma_X(b)$ \\
\hline
$- \sigma_Z \otimes \sigma_X(b)$ & $- \sigma_X \otimes \sigma_Z(a)$ & $- \sigma_Z \sigma_X \otimes \sigma_Z(a) \sigma_X(b)$
\end{tabular}

\emph{Operators before the tensor product are understood always to act on the first qubit of Alice's halves of the shared EPR pairs, and operators after the tensor product on the last $n$ qubits.}
\item If $j = 2$ (i.e.~Bob's question indicates the top centre cell), the verifier sends $Z$ to Bob, receives $v \in \{0,1\}^n$ as an answer, and accepts iff $\prod_{i: a_i = 1} v_{i}$ is equal to $u_2$ (if Alice was asked a row question) or $u_1$ (if Alice was asked a column question). If $j = 4$ (i.e.~Bob's question indicates the centre left cell), the verifier sends $X$ to Bob, receives $v \in \{0,1\}^n$ as an answer, and accepts iff $\prod_{i: b_i = 1} v_{i}$ is equal to $u_1$ (if Alice was asked a row question) or $u_2$ (if Alice was asked a column question. In all other cases, the verifier sends $(\mathsf{MS}, a, b, j)$ to Bob, receives a single-bit answer $v \in \{0,1\}$, and accepts iff $v = u_k$ for the $k \in [3]$ such that $u_k = j$. \emph{Honest Bob measures all his qubits in the $W$ basis when he receives a single bit question $W$, and in all other cases uses the same strategy as honest Alice, except that he only measures a single cell instead of 3.}

\end{enumerate}
\end{longfbox}

\begin{longfbox}[breakable=false, padding=1em, padding-right=1.8em, padding-top=1.2em, margin-top=1em, margin-bottom=1em]
\begin{protocol} Mixed-versus-pure basis test 
\label{prot:mixed-vs-pure} \end{protocol}
\textbf{Input:} A distribution $D$ over $\{\1, X,Z\}^n$ which can be sampled efficiently. 

\begin{enumerate}
    \item The verifier selects a basis $W$ from $\{X,Z\}$ uniformly at random and sends $W$ to Alice. It receives answer $u \in \{0,1\}^n$. \emph{Honest Alice measures all of her qubits in the $W$ basis.}
    \item The verifier selects a uniformly random $b \gets \{0,1\}$.
    \begin{enumerate}
        \item If $b=0$, the verifier sends $W$ to Bob, and receives answer $v \in \{0,1\}^n$. It accepts iff $u=v$. \emph{Honest Bob measures all of his qubits in the $W$ basis.}
        \item If $b=1$, the verifier samples a Pauli string $w$ from $D$ and sends it 
        to Bob. It accepts iff for all $i$ where $w_i = W$, it is the case that $u_i = v_i$. \emph{Honest Bob measures his last $n$ qubits in the Pauli bases designated by $w$ and reports all the measurement results. For qubits where $w_i = \1$, he always reports the measurement result 0.}
    \end{enumerate}
\end{enumerate}
\end{longfbox}

\begin{longfbox}[breakable=false, padding=1em, padding-right=1.8em, padding-top=1.2em, margin-top=1em, margin-bottom=1em]
\begin{protocol} Hamiltonian test
\label{prot:hamiltonian} \end{protocol}
\textbf{Input:} a distribution $D$ over $\{\1, X,Z\}^n$ 
along with a collection of sets $\{Q(w) : w \in \{\1,X,Z\}^n\}$ such that $Q(w) \subseteq \{0,1\}^n$ (described in a way such that membership in $Q(w)$ can be efficiently decided given $w$).
\begin{enumerate}
    \item The verifier sends question $\mathsf{tele}$ to Alice, and receives in response two strings, $u_x, u_z \in \{0,1\}^n$. \emph{Honest Alice teleports her ground state to Bob through their last $n$ shared EPR pairs and reports the teleportation corrections.}
    \item The verifier samples a Pauli string $w$. 
    and sends it to Bob. 
    The verifier then receives measurement outcomes from Bob, corrects Bob's outcomes using Alice's reported teleportation corrections, and does the appropriate energy test. More specifically, the verifier receives answer $v$ from Bob, and computes for every $i$ such that $w_i \neq \1$
\begin{equation*}
s_i = \underbrace{v_i}_{\text{Bob's measurement}} \oplus \underbrace{[(u_x)_i]^{\1[w_i = Z]} \oplus [(u_z)_i]^{\1[w_i = X]}}_{\text{correction from Alice}}
\end{equation*}
Then the verifier sets $s_i = 0$ for every $i$ such that $w_i = \1$, and accepts iff $s \in Q(w)$. \emph{Honest Bob measures his last $n$ qubits in the Pauli bases designated by $w$ and reports all the measurement results. For qubits where $w_i = \1$, he always reports the measurement result 0.}
\end{enumerate}
\end{longfbox}

\subsection{Applying the Generalized KLVY Compiler
\label{sec:compiled_nonsuccinct_qma}}

We next compile Protocol~\ref{prot:main-2-prover} into a single-prover argument system for $\QMA$ (Protocol~\ref{prot:main-blind}), by applying the generalized KLVY compiler (\cref{def:generalized_KLVY-compiler}). Recall that in the compiled protocol, the verifier interacts with the prover first through a blind delegation protocol satisfying \cref{def:blind-computation}, where the client samples and delegates Alice's question $x$, and the server blindly runs Alice's operation on $x$ to obtain outcome $a$. Next, the verifier sends Bob's question to the prover (in the clear) and receives the response. The verifier will decide the outcome based on the questions $(x,y)$ and the prover's responses $(a,b)$.




\begin{longfbox}
\begin{protocol} Classically-verifiable argument system for $\QMA$ 
\label{prot:main-blind} \end{protocol}
\textbf{Inputs:} An instance $x \in \{0,1\}^*$ and an algorithm $C$ which is the verifier for a $\QMA$ promise problem $\cL = (\cL_\mathrm{yes}, \cL_\mathrm{no})$ such that $\cL \in \mathsf{QMA}$. A security parameter $\secp \in \N$.

The verifier executes Protocol~\ref{prot:main-2-prover} under the compiler given in Definition~\ref{def:generalized_KLVY-compiler}. 
\end{longfbox}

\subsection{Prover Modelling and Useful Lemmas} 
\label{sec:qma_useful_lemma}

\paragraph{Alice's circuit depth.} To ensure that protocol ~\ref{prot:main-blind} has bounded round complexity, we need to make sure that the circuit depth of the operation performed by Alice in ~\ref{prot:main-2-prover} has bounded dependence on the input (instance) size. Indeed, the blind delegation protocol between the client and the server (Alice part) has round complexity that depends polynomially on the quantum circuit depth for the delegated Alice computation.\footnote{We refer the reader to \cite[Section 6.3]{bartusek2025power} for how the blind delegation works.}

\begin{claim}   \label{claim:alice_low_circuit_depth}
The circuit performed by honest Alice in \cref{prot:main-blind} has depth $\poly(\log(|x|))$. %
\end{claim}

\begin{proof}
Recall that the parameter $n$ in \cite{metger2024succinct} is such that $n := \poly(|x|)$. We will discuss the dependence in terms of $n$. We can obtain the claim by looking into the subprotocols of ~\ref{prot:main-2-prover}: the Pauli braiding test, mixed-versus-pure basis test, and Hamiltonian test. In particular, we look at Step 1 of ~\ref{prot:comgame}, Step 2 of ~\ref{prot:acomgame} and Step 1 of ~\ref{prot:hamiltonian}. The only operations Alice performs are measurements in the Pauli basis on an $O(n)$-qubit state, with each qubit measured \emph{in parallel} in the subtests ~\ref{prot:comgame} and ~\ref{prot:acomgame}, then returning an XOR of the $n$ measurement results. This can be done using a circuit of depth $\log(n) < \log^2(|x|)$. In ~\ref{prot:hamiltonian}, she performs a measurement
  in the Bell basis on an $O(n)$-qubit state (with each 2-qubit pair measured in parallel), and thus only needs a $O(1)$ circuit depth. %
\end{proof}

\begin{remark}
    Note that 
the communication size remains polynomially dependent on $n$, as we can see from the size of the challenger's messages in ~\ref{prot:pauli-braiding}, ~\ref{prot:mixed-vs-pure}, and ~\ref{prot:hamiltonian}. This is later made succinct using our communication compression compiler.
\end{remark}

\paragraph{Notations and modelling of the adversary.} We use a few notations from \cite{metger2024succinct} and \cite{bartusek2025power}: %

\begin{itemize}
    \item Let $\ket{\psi}$ be the initial internal state of the prover over the system $\cH_{\cA} \otimes \cH_{\cB}$.
    For simplcity, we denote $\psi$ for $\ket{\psi}\bra{\psi}$.
    
    \item Let $A_a^x$ be Alice's measurement operator corresponding to question $x$ and response $a$. We can write $A_{a}^x = \ket{a}\bra{a}U^x$, where $U^x$ is an arbitrary unitary acting on 
    $\cH_{A} \otimes \cH_{B}$. Note that the collection of operators $\{(A^x_a)^\dagger (A^x_a)\}_a$ forms a projective measurement.  Thus, the probability that Alice returns outcome $a$ in response to question $x$ is
    \[ \Pr[a] = \bra{\psi} (A^x_a)^\dagger (A^x_a) \ket{\psi}. \]

     \item The \emph{un-normalized post-measurement state} after receiving question $x$ and responding with answer $a$ is
    \begin{equation} \ket{\psi^{x}_{a}} = A^x_a \ket{\psi}. \label{eq:def-post-meas-states-chsh} \end{equation}
    Note that $\| \ket{\psi^x_a}\|^2 = \Pr[a]$.
    The post-measurement state marginalizing over outcomes for question $x$ is the mixed state
    \[ \psi^x = \sum_a \psi^x_a = \sum_a\ket{\psi^x_a}\bra{\psi^x_a}. \]
    \item The second measurement (the ``Bob measurement") is modeled by a POVM $\{B^y_b\}_y$ for each question $y$. While this measurement can always be taken to be projective without loss of generality, in the analysis it will sometimes be convenient to construct strategies where this measurement is a POVM.
\end{itemize}

Before going into the modified theorems/proofs for \cite{metger2024succinct}, we first give a simple lemma important for the proofs in \cite{metger2024succinct}, adapted to our setting. In \cite{metger2024succinct}, this follows from the security of the QFHE. In our scenario, it follows from the blindness property of the blind delegation.

\begin{lemma}
\label{lem:ind_alice_state}
For any two Alice questions $x_1,x_2 \in \cQ_A$,
the post-measurement states
$\psi^{x_1}$ and $\psi^{x_2}$ are computationally indistinguishable. 

\end{lemma}

\begin{proof}
This follows directly from the security of the blind delegation protocol defined in ~\ref{def:blind-computation}. By the end of the interaction between the verifier and the ``Alice'' stage of the prover,  no QPT distinguisher (which is the prover itself plus outputing a bit at the end to indicate which input it thinks it receives) can distinguish between any two verifier (client)'s inputs (questions). This means that the prover state $\psi^x$ at this point must satisfy the property that $\psi^{x_1} \approx_{\negl} \psi^{x_2}$ for any $x_1, x_2$, otherwise it could break the security of blind delegation.
\end{proof}





\paragraph{Measuring closeness of strategies.}
In the analysis of the compiled non-local game in both our setting and that of \cite{metger2024succinct}, it happens that if we replace  Bob's measurements $\{B^y_b\}$ with new measurements $\{C^y_b\}$ that are close in the appropriate distance metric, then the winning probability of the strategy is approximately preserved. Specifically, we will often measure closeness in terms of the state-dependent norm on the post-measurement state, after Alice's measurement has been applied:
\[ \E_{x, y \sim D_G}  \sum_{a, y} \| B^y_b - C^y_b \|_{\psi^{x}_a}^2 \leq \eps, \]
where $D_G$ is the distribution over questions sampled in the game.


\begin{lemma}[Lemma 5.3 from \cite{metger2024succinct}]
  \label{lem:close_strats}
    Let $G$ be a nonlocal game, $\ket{\psi}$ a state, $\{A^x_a\}$ be a collection of compiled Alice measurements, and $\{B^y_b\}$ and $\{C^y_b\}$  be two collections of Bob measurements such that $B^y_b$ is projective, and 
    \[ \E_{x,y \sim D_G}  \sum_\alpha \sum_b  \| B^y_b - C^y_b \|_{\psi^{x}_a}^2 \leq \eps,\]
    where $D_G$ is the distribution over question pairs in the game.
    Then the success probabilities of the strategy using $(\ket{\psi}, \{A^x_a\}, \{B^y_b\})$ and the strategy using $(\ket{\psi}, \{A^c_\alpha\}, \{C^y_b\})$ are $O(\sqrt{\eps})$-close. 
\end{lemma}



\subsection{Analysis of the Compiled Protocol}
We now show that the soundness of protocol ~\ref{prot:main-blind} holds (according to protocol ~\ref{def:CVQC}), by adapting the theorems in \cite{metger2024succinct}.

\paragraph{Compiled Pauli-braiding test.} We first show that we can rephrase the proofs in \cite{metger2024succinct} on the rigidity of the QFHE compiled Pauli-braiding test into that for the generalized KLVY compiled protocol.

We restate the  following lemma from \cite{metger2024succinct} for our setting.
\begin{lemma}\label{lem:pauli-braiding-klvy}
Suppose that a QPT prover $P$ as modeled in ~\ref{sec:qma_useful_lemma} wins with probability $1 - \eps$ in the generalized KLVY compiled version of protocol ~\ref{prot:pauli-braiding}.
Then the prover's observables satisfy that for any Alice question $q \in \cQ_A$:
\begin{align}
\E_{a, b \in \bits^n} \norm{Z(a) X(b) - (-1)^{a \cdot b} X(b) Z(a)}_{\psi^q}^2 \leq O(\eps) + \negl(\lambda) \,. \label{eqn:unif_commutation}
\end{align}
\end{lemma}

To see this, we will first note the following, which is a rephrase of Lemma 6.4 from \cite{metger2024succinct}.
\begin{lemma}
\label{thm:comp_nonlocal_acomgame}
 Suppose that a QPT prover $P$ as modeled in \cref{sec:qma_useful_lemma} succeeds with probability $1-\eps$ in the generalized KLVY compiled anti-commutation test (protocol ~\ref{prot:acomgame}) with inputs $a, b$, and let $B^{a}, B^{b}$ be Bob's observables corresponding to the questions $a, b$. That is, $B^{a} \coloneqq B^a_0 - B^a_1, B^b \coloneqq B^b_0 - B^b_1$. Then for any Alice question $x \in \cQ_A$, there exists a negligible function $\eta(\lambda)$ (depending on $P$ and on $x, a, b$) such that
    \[ \| \{ B^{a}, B^{b}\} \|_{\psi^{x}}^2 \leq O(\eps) + \eta(\lambda)\,.\]  
\end{lemma}
\begin{proof}
    This follows via essentially the same proof as of \cite{metger2024succinct} Lemma 6.4. In our case, Bob's observable takes as input the full $n$-bit strings $a,b$ instead of the sampling randomness $r_a, r_b$ used in \cite{metger2024succinct}, though this change has no affect on the proof.
\end{proof}



Next, we note the following, which is a rephrase of Lemma 6.5 from \cite{metger2024succinct}.

\begin{lemma}
\label{lem:com-test}
    Suppose that a QPT prover $P$ as modeled in \cref{sec:qma_useful_lemma} succeeds with probability $1-\eps$ in the generalized KLVY compiled commutation test (protocol ~\ref{prot:comgame}) with inputs $a, b$, and let $B^{a}, B^{b}$ be Bob's observables corresponding to the questions $a, b$. Then for any Alice question $x \in \cQ_A$, there exists a negligible function $\eta(\lambda)$ (depending on $P$ and on $x, a, b$) such that
    \begin{align*}
        \Vert [ B^{a}, B^{b}] \|_{\psi^{x}}^2 \leq O(\eps) + \eta(\lambda)\,.
    \end{align*}
\end{lemma}

\begin{proof}
    As in the above lemma, this follows directly from the original proof of \cite[Lemma 6.5]{metger2024succinct}.
\end{proof}

Now we can give the proof for \cref{lem:pauli-braiding-klvy}. 
\begin{proof}[Proof of \cref{lem:pauli-braiding-klvy}]
The proof mostly follows from Lemma 6.3 in \cite{metger2024succinct}. The only difference is that we do not have to invoke PRG security since we do not use the PRG for sampling the questions $a,b$ and instead send them directly.

We are given that $P$ succeeds with probability $1 - \eps$ in protocol ~\ref{prot:pauli-braiding}. This means that it holds that
\begin{align*}
    &\sum_{a,b \in \{0,1\}^n : a \cdot b = 0} \frac{1}{2^{2n}} \cdot (1 - \Pr[P\text{ passes commutation on }Z(a), X(b)])\nonumber \\
    &\qquad +  \sum_{a,b \in \{0,1\}^n : a \cdot b = 1} \frac{1}{2^{2n}} \cdot (1-\Pr[P\text{ passes anticommutation on }Z(a), X(b)]) \nonumber \\
    &\qquad \leq \eps.
\end{align*}
For any given $a,b$, suppose $a \cdot b = 0$, and let $1-\eps_{a,b} =\Pr[P\text{ passes commutation on }Z(a), X(b)]$. Then by \cref{lem:com-test}, it holds that
\[ \| [Z(a), X(b) ] \|_{\psi^{q}}^2 \leq O(\eps_{a,b}) + \eta_{a,b}(\lambda), \]
where $\eta_{a,b}(\cdot)$ is a negligible function.

Likewise, suppose $a \cdot b = 1$, and let $1-\eps_{a,b} =\Pr[P\text{ passes anticommutation on }Z(a), X(b)]$. Then by \cref{thm:comp_nonlocal_acomgame}, it holds that
\[ \| \{Z(a), X(b) \} \|_{\psi^{q}}^2 \leq O(\eps_{a,b}) + \eta_{a,b}(\lambda), \]
where $\eta_{a,b}(\cdot)$ is a negligible function.

Putting these relations together, we have
\[ \sum_{a,b \in \{0,1\}^n} \frac{1}{2^{2n}} \| Z(a) X(b) - (-1)^{a \cdot b} X(b)Z(a) \|_{\psi^{q}}^2 \leq O(\eps) + \eta(\lambda),\]
where \[\eta(\secp) \coloneqq \frac{1}{2^{2n}}\sum_{a,b \in \{0,1\}^n}\eta_{a,b}(\secp) = \negl(\secp).\] 


\end{proof}

\paragraph{Compiled mixed-vs-pure basis test.}
We next discuss the analysis of the generalized KLVY compiled protocol ~\ref{prot:mixed-vs-pure} and how we adapt the analysis from \cite{metger2024succinct} analysis to our case.
\begin{theorem} \label{thm:mixed_basis_final}
Suppose that a QPT prover $P$ wins with probability $1 - \eps$ in both protocol ~\ref{prot:pauli-braiding} and protocol ~\ref{prot:mixed-vs-pure}. Then there exists a Hilbert space $\cH' = \cH_{A} \otimes \cH_{B}$
(we can also represent the above Hilbert space as $ \C^{2^n} \otimes \cH_{\rm aux}$ where $n$ is the number of qubits an honest prover would use and $\cH_{\rm aux}$ is an arbitrarily polynomial size auxiliary register used by the malicious prover) and an isometry $V: \cH \to \cH'$ such that for every Alice question $q \in \cQ_A$
\begin{align*}
\E_{w \sim D} \E_{a \in \bits^n} \| O^w(a) - V^\dagger (\sigma_w(a) \otimes \1_{\rm aux}) V \|_{\psi^{q}}^2 \leq O(\sqrt\eps) + \negl(\lambda) \,.
\end{align*}
\end{theorem}

\cref{thm:mixed_basis_final} also implies that the prover's projective measurements $\{M^w_u\}$ in the mixed vs pure basis test must be close to the corresponding Pauli projectors, as shown in \cite{metger2024succinct}. We will omit the statement and proof of this corollary (Lemma 6.9 in \cite{metger2024succinct}) since it follows via minor notational changes: Replace $\psi^{\Enc(q)}$ with our post-Alice-measurement state $\psi^{q}$ and $\Dec(\alpha)$ with $a$ (Alice's answer in the clear) for both the statement and the proof. 

For the proof of \cref{thm:mixed_basis_final}, we can also follow the proof in \cite{metger2024succinct} except for the notational changes mentioned above. The reason is the following: the only computational property \cite{metger2024succinct} uses is the computational indistinguishability of Alice's post-measurement states under different questions (used in \cite{metger2024succinct} Lemma 6.11), which is also true for our Alice post-measurement states as stated in \cref{lem:ind_alice_state}. The rest of the analysis is information theoretic, which applies exactly the same to our setting and we omit it here.

\paragraph{Compiled Hamiltonian test.}
This part of the analysis for our protocol also follows directly from the proof of \cite[Lemma 6.22]{metger2024succinct}. We state the theorem here for the sake of a complete presentation.
Note that our protocol protocol ~\ref{prot:main-2-prover} (and thus protocol ~\ref{prot:hamiltonian}) has made the modification that the verifier does not use a ``subsampled'' Hamiltonian that requires only succinct-length randomness (as in \cite[Lemma 6.20]{metger2024succinct}) but rather ``full'' randomness as defined in \cref{thm:qma-to-hamiltonian}. This change does not result in changes in the proof of the following theorem (except for the notational changes same as above), since the Hamiltonian problem properties do not change. 

\begin{lemma}\label{lem:ham-test-analysis}
    Given an instance $x$, let $H(n)$ be the corresponding Hamiltonian generated in protocol ~\ref{prot:main-2-prover}, and let $P$ be a prover that succeeds in the compiled Pauli braiding and mixed-versus-pure tests (protocol ~\ref{prot:pauli-braiding} and protocol ~\ref{prot:mixed-vs-pure}) with probability $1-\delta$, and in the Hamiltonian test (protocol ~\ref{prot:hamiltonian}) with probability $p_{H}$. Then there exists a state $\rho$ such that
    \[ \Tr[\rho H(n)] \leq (1 - p_{H}) + O(\delta^{1/4}) + \negl(n).\]
\end{lemma}

\paragraph{Putting it together.} Now, we put our arguments together to  conclude the following theorem for the compiled protocol under blind delegation. 


\begin{theorem}
\label{thm:round_efficient_classical_qma}
    Assuming OSP, there exists a constant $s$ such that protocol ~\ref{prot:main-blind} is a (non-succinct) classically-verifiable $(1-\negl(\secp),s)$ argument system for QMA (\cref{def:CVQC}). Moreover, the number of rounds in the protocol is bounded by $\poly(\secp)$ for some fixed polynomial $\poly$. %

    

\end{theorem}
The OSP assumption is due to the use of blind delegation from OSP as stated in \cref{thm:osp_imply_blind}.

The completeness property follows from the correctness of the blind delegation (\cref{def:blind-computation}) as well as the completeness of the protocol from \cite{metger2024succinct}. The soundness proof follows that of \cite{metger2024succinct} by combining the above lemmas on each subtest. We do not need to make any changes and therefore omit the full proof here. We can conclude the round complexity by combining \cref{claim:blind_depend_depth} and \cref{claim:alice_low_circuit_depth}, as $|x|$ is always polynomially-related to the security parameter $\secp$. 


%% file: hamiltonian_prelims.tex


\subsubsection{Hamiltonian problems and amplification}
\label{sec:hamiltonian-subsamp}
We will provide definitions and theorems regarding the Hamiltonian problem used in \cref{prot:hamiltonian}.

\paragraph{Difference from \cite{metger2024succinct}:} The main difference (and simplification) of our case from \cite{metger2024succinct} is that we do not need the ``subsampled Hamiltonian'' sampled using a PRG, since we do not need the verifier to be question-succinct at this point. We therefore remove the need for \cite[Lemma 6.19]{metger2024succinct} and \cite[Lemma 6.20]{metger2024succinct} from this section.

\begin{definition}[Hamiltonian problem]
\label{def:hamiltonian-problem}
We refer to a tuple of the form $(H, \alpha, \beta)$, where $H$ is a Hermitian operator and $\alpha$ and $\beta$ are both real numbers, as a \emph{Hamiltonian problem}. We may refer to such a tuple where $H$ acts on $n$ qubits as an $n$-qubit Hamiltonian problem.
\end{definition}

\begin{definition}[Deciding a Hamiltonian problem]
Given a tuple of the form $(H, \alpha, \beta)$, where $H$ is a Hermitian operator and $\alpha$ and $\beta$ are both real numbers, we refer to the problem of deciding whether the ground energy of $H$ is $\leq \alpha$ (yes case) or $\geq \beta$ (no case) as the problem of \emph{deciding} $(H, \alpha, \beta)$.
\end{definition}

\begin{definition}[Family of Hamiltonians]
We will use the notation $\cH = \{\cH(n)\}_{n \in \mathbb{N}}$ to refer to a family of Hamiltonians, i.e., a collection of sets $\cH(n)$ (one for every value of $n \in \mathbb{N}$) such that the $n$th set $\cH(n)$ contains only $n$-qubit Hamiltonian problems. We assume that the length of the binary description of $(H, \alpha, \beta)$ for any $(H, \alpha, \beta) \in \cH(n)$ is $\poly(n)$.
\end{definition}

\begin{definition}[QMA-completeness of a family of Hamiltonians]
We say a family of Hamiltonians $\cH = \{\cH(n)\}_{n \in \mathbb{N}}$ is QMA-complete if, for every promise problem $\cL = (\cL_\mathrm{yes}, \cL_\mathrm{no})$ in $\mathsf{QMA}$, and every instance $x \in \{0,1\}^*$, the problem of deciding whether $x \in \cL_\mathrm{yes}$ or $x \in \cL_\mathrm{no}$ can be Karp reduced in polynomial time (in $|x|$, given as input $x$ and also the algorithm $C$ which characterizes the verifier for $\cL$) to the problem of deciding some element of $\cH(n)$ for $n = \poly(|x|)$.
\end{definition}

\begin{lemma}[\cite{metger2024succinct} Lemma 6.18]\label{thm:qma-to-hamiltonian}
Let $C$ be a $\QMA$ verification algorithm for a $\QMA$ promise problem $\cL = (\cL_\mathrm{yes}, \cL_\mathrm{no})$, and let $x$ be an input of length $n$. Then there is a Hamiltonian problem $(H, \alpha(n), \beta(n))$ that can be efficiently computed from the input $(x, C)$, such that if $x \in \cL_\mathrm{yes}$, then $H$ is a YES instance of the Hamiltonian problem, and if $x \in \cL_\mathrm{no}$, then $H$ is a NO instance, and $\beta(n) - \alpha(n) = 1 - \negl(n)$. Moreover, $H$ has the following form:
\begin{align*}
H = \1 - \underset{w \sim D}{\E}\sum_{a \in T(w)} \pi^w_a,
\end{align*}
where $D$ is a distribution over $\{\1,X,Z\}^{\poly(n)}$ which can be efficiently sampled, and $T(w) \subseteq \{0,1\}^{\poly(n)}$ is a set for which membership can be decided in time polynomial in $n$.
\end{lemma}

%% file: blind_delegation_nonlocal.tex
\subsubsection{Generalized KLVY compiler}
\label{sec:blind_delegation}

In this section, we state a few definitions and theorems from
\cite{bartusek2025power} regarding a compiler that utilizes classical-client blind delegation of quantum computation to compile (classical-verifier) two-prover non-local games to single-prover protocols.


The seminal work of \cite{kalai2023quantum} presented a method for converting two-prover protocols into single-prover protocols, using quantum FHE. They showed quantum correctness and classical soundness of the compiled protocol, as a new approach for proof of quantumness. A sequence of follow-ups, including \cite{natarajan2023bounding,metger2024succinct}, have demonstrated the quantum soundness of the compiler in various settings. \cite{bartusek2025power} further showed that the quantum FHE used in this compiler can in fact be replaced with blind delegation of quantum computation. 

\paragraph{Blind Classical Delegation of Quantum Computation.}
We first give a very general definition for blind delegation of quantum computation with a classical client. It allows the client to delegate the computation of an arbitrary publicly-known quantum operation that takes a quantum input (provided by the server, and potentially entangled with an auxiliary system held by the server) and a private classical input (chosen by the client). After interaction, the server is able to obtain the (potentially quantum) output up to a one-time pad with keys known to the client.


\begin{definition}[Blind Classical Delegation of Quantum Computation]\label{def:blind-computation}
Let $\cH_V,\cH_W$ be Hilbert spaces of arbitrary dimension, and let $Q: \{0,1\}^* \times \cH_V \to \cH_W$ be a polynomial-size quantum circuit that takes as input a classical string $x$ and a state on register $V$, and outputs a state on register $W$. A protocol for blind classical delegation of quantum computation consists of an interaction 

\[(W,(r,s)) \gets \langle S(1^\secp,Q,V),C(1^\secp,Q,x)\rangle\]

between 

\begin{itemize}
    \item a QPT server $S(1^\secp,Q,V)$ with input the security parameter $1^\secp$, circuit $Q$, and state on register $V$, and
    \item a PPT client $C(1^\secp,Q,x)$ with input the security parameter $1^\secp$, circuit $Q$, and classical string $x$.
\end{itemize}

At the end of the interaction, the server outputs a state on register $W$ and the client outputs classical strings $(r,s)$. The protocol must satisfy the following properties.
\begin{itemize}
    \item \textbf{Correctness.} For any circuit $Q$ and input $x$, let $\mathsf{IDEAL}[Q,x]$ be the map from $V \to W$ defined by $Q(x,\cdot)$, and let $\mathsf{REAL}[Q,x]_\secp$ be the map from $V \to W$ induced by running the protocol \[(W,(r,s)) \gets \langle S(1^\secp,Q,V),C(1^\secp,Q,x)\rangle\] and then applying $X^r Z^s$ to $W$. Then for any $Q$ and $x$,
    \[D_\diamond\left(\mathsf{REAL}[Q,x]_\secp, \mathsf{IDEAL}[Q,x]\right) = \negl(\secp).\]



    \item \textbf{Security.} For any circuit $Q$, QPT adversary $\{\Adv_\secp\}_{\secp \in \bbN}$, and two strings $x_0,x_1$, it holds that 
    \begin{align*}&\bigg| \Pr\left[b_\Adv = 0 :  (b_\Adv,(r,s)) \gets \langle \Adv_\secp,C(1^\secp,Q,x_0)\rangle\right] \\ &-\Pr\left[b_\Adv = 0 :  (b_\Adv,(r,s)) \gets \langle \Adv_\secp,C(1^\secp,Q,x_1)\rangle\right] \bigg| = \negl(\secp), 
    \end{align*}
    where $b_\Adv$ denotes a single bit output by $\Adv_\secp$ after the interaction (which could result from an arbitrary QPT operation applied to its state after the interaction).
\end{itemize}
\end{definition}

\begin{remark}\label{remark:classical-output}
    Note that the above definition implies that if $Q$ has a \emph{classical} output, then the client can obtain this output with one extra message from the server. That is, suppose $Q : \{0,1\}^* \times \cH_V \to \{0,1\}^*$. Then at the conclusion of the protocol defined above, the server has a classical output $y \oplus r$, and the client has the one-time pad key $r$ (note that $s$ is irrelevant once the output has been measured in the standard basis). Then, the server returns $y \oplus r$ to the client, who recovers the output $y$. In this case, we denote the protocol 
    \[y \gets \langle S(1^\secp,Q,V),C(1^\secp,Q,x)\rangle,\] where $y$ is the client's output.
\end{remark}

\paragraph{OSP implies Blind Classical Delegation.} We will use the following theorem from \cite{bartusek2025power}, along with the following claim.
\begin{theorem}
\label{thm:osp_imply_blind}
    OSP (\cref{def:OSP}) implies blind classical delegation of quantum computation (\cref{def:blind-computation}).
\end{theorem}

\begin{claim}
\label{claim:blind_depend_depth}
    The round complexity of the blind classical delegation of quantum computation protocol from \cite{bartusek2025power} is $\poly(D, \lambda)$, where $D$ is the depth of the quantum circuit $Q$ (in fact, the $T$-depth) and $\lambda$ is the security parameter.
\end{claim}
This is a simple observation that follows by inspecting the protocol from \cite[Section 6.3]{bartusek2025power}. Since we do not use the protocol in any non-black-box manner, we omit the details here.

\paragraph{Generalized KLVY Compiler.} Finally, we state some definitions, followed by the statement of the generalized KLVY compiler, where in place of a QFHE protocol we use any blind classical delegation of quantum computation protocol (\cref{def:blind-computation}).

\begin{definition}[Nonlocal game]
    A \emph{nonlocal game} $G = (Q,V)$ is specified by a distribution $Q$ over pairs $(x,y) \in \{0,1\}^{n_1} \times \{0,1\}^{n_2}$ and a verification predicate $V(x,y,a,b) \in \{0,1\}$, where $a \in \{0,1\}^{m_1}$ and $b \in \{0,1\}^{m_2}$. A \emph{family} of nonlocal games $\cG = \{\cG_\secp\}_{\secp \in \bbN}$ is a set of games parameterized by $\secp$, where each $\cG_\secp$ is itself a set of games $G \in \cG_\secp$. Each game $G \in \cG$ is defined by a distribution $Q_G$ over pairs $(x,y) \in \{0,1\}^{n_{1,G}} \times \{0,1\}^{n_{2,G}}$ and a verification predicate $V_G(x,y,a,b) \in \{0,1\}$, where $a \in \{0,1\}^{m_{1,G}}$ and $b \in \{0,1\}^{m_{2,G}}$. For any game $G \in \cG$, we define $\secp_G$ to be the choice of $\secp$ such that $G \in \cG_\secp$. We say that the family of games is \emph{efficient} if there exists a polynomial $p(\cdot)$ and a procedure that, for any $G \in \cG$, samples from $Q_G$ and computes $V_G$ in time $p(\secp_G)$.
\end{definition}

\begin{definition}[Non-local Strategy]
\label{def:nonlocal_strategy}
 A \emph{nonlocal strategy} $\mathscr{S}$ for game $G$ consists of the following.
    \begin{itemize}
        \item A state $\ket{\psi} \in \cH_\cA \otimes \cH_\cB$.
        \item For every $x \in \{0,1\}^{n_1}$, a projective measurement $\{A_a^x\}_{a}$ acting on $\cH_A$ with outcomes $a \in \{0,1\}^{m_1}$.
        \item For every $y \in \{0,1\}^{n_2}$, a projective measurement $\{B_b^y\}_{b}$ acting on $\cH_B$ with outcomes $b \in \{0,1\}^{m_2}$.
    \end{itemize}
    The \emph{value} of this strategy is given by 
    \[\omega(G,\mathscr{S}) \coloneqq \E_{(x,y) \gets Q}\sum_{a,b} V(x,y,a,b) \bra{\psi}A_a^x \otimes B_b^y\ket{\psi}.\]


    A strategy $\mathscr{S}$ for a \emph{family} of games $\cG$ consists of a strategy

    \[\mathscr{S}_G = \left(\ket{\psi_G},\{A^x_{a,G}\}_a,\{B^y_{b,G}\}_b\right)\]

    for each $G \in \cG$. We say that $\mathscr{S}$ is \emph{efficient} if $\ket{\psi}$ is QPT-preparable and $\{A_{a,G}^x\}_a$ and $\{B_{b,G}^y\}_b$ are QPT-implementable.
    
\end{definition}

\begin{definition}[Generalized KLVY Compiler]
\label{def:generalized_KLVY-compiler}
    Let $\cG$ be a family of nonlocal games, let $\Pi = \langle S,C \rangle$ be a blind classical delegation of quantum computation protocol, and let $\mathscr{S}$ be an efficient nonlocal strategy for $\cG$. For each $G \in \cG$, let $A_G : \{0,1\}^{n_{1,G}} \times \cH_{\cA} \to \{0,1\}^{m_{1,G}}$ be the QPT circuit that performs the Alice measurement of $\mathscr{S}_G$ and $B_G : \{0,1\}^{n_{2,G}} \times \cH_{\cB} \to \{0,1\}^{m_{2,G}}$ be the QPT circuit that performs the Bob measurement of $\mathscr{S}_G$. The \emph{KLVY-compiled} protocol $\KLVY[\cG,\Pi,\mathscr{S}]$ is an interaction between a QPT prover $\Prove$ and a PPT verifier $\Ver$ that is parameterized by a game $G \in \cG_\secp \subset \cG$, and operates as follows.
    \begin{enumerate}
        \item The verifier samples $(x,y) \gets Q_G$.
        \item Let $\ket{\psi_G} \in \cH_{\cA} \otimes \cH_{\cB}$ be the initial state used in $\mathscr{S}_G$. The prover and verifier engage in a blind classical delegation of quantum computation protocol (with classical output, see \cref{remark:classical-output}) \[a \gets \langle S(1^\secp,A_G,\ket{\psi_G}), C(1^\secp,A_G,x)\rangle,\] with the prover playing the role of the server $S$ and the verifier playing the role of the client $C$.
        \item The verifier sends $y$ to the prover.
        \item The prover runs $b \gets B_G(y,\cB)$ and sends $b$ to the verifier.
        \item The verifier outputs $V_G(x,y,a,b)$.
    \end{enumerate}
    This interaction is denoted \[b_\Ver \gets \langle \Prove(1^\secp,G),\Ver(1^\secp,G)\rangle,\] where $b_\Ver$ denotes the bit output by $\Ver$.

    The \emph{completeness} of $\KLVY[\cG,\Pi,\mathscr{S}]$ is defined by a function \[c(G) \coloneqq \Pr\left[b_\Ver = 1 : b_\Ver \gets \langle \Prove(1^\secp,G),\Ver(1^\secp,G)\rangle\right],\] and $\KLVY[\cG,\Pi,\mathscr{S}]$ has soundness $s(G)$ if for any QPT adversary $\{\Adv_\secp\}_{\secp \in \bbN}$
    \[\Pr\left[b_\Ver = 1 : b_\Ver \gets \langle \Adv_{\secp_G}(1^\secp,G),\Ver(1^\secp,G)\rangle\right] \leq s(G).\]
    
\end{definition}